\documentclass[sigconf]{acmart}

\usepackage{amsmath,amsfonts}
\usepackage[noend]{algpseudocode}
\usepackage{graphicx}
\usepackage{textcomp}
\usepackage{xspace}
\usepackage{url}
\usepackage[utf8]{inputenc}
\usepackage{microtype}
\usepackage{graphicx,mathtools,kantlipsum}
\usepackage[normalem]{ulem}
\usepackage{amsthm}
\usepackage{listings}
\usepackage{courier}
\usepackage{float}
\usepackage{multirow}
\usepackage{makecell}
\usepackage{booktabs}
\usepackage[font=small]{caption}
\usepackage{subcaption}
\usepackage{hyperref}
\usepackage[capitalize]{cleveref}
\usepackage{balance}
\usepackage{pifont}
\usepackage{enumitem}
\usepackage{textcomp}
\usepackage{comment}
\usepackage{outlines}
\usepackage{algorithm}

\usepackage{tikz}
\usetikzlibrary{fit,calc}
\usepackage[hang,flushmargin]{footmisc}

\algblock{Input}{EndInput}
\algnotext{EndInput}
\algblock{Output}{EndOutput}
\algnotext{EndOutput}

\usepackage{xcolor}
\definecolor{codegreen}{rgb}{0,0.4,0}
\hypersetup{
    colorlinks=true,
    citecolor=blue,
    linkcolor=blue,
}

\newlength{\mysep}
\setlist[enumerate]{noitemsep,topsep=0pt,partopsep=3pt,leftmargin=*}
\setlist[itemize]{noitemsep,topsep=0pt,partopsep=3pt,leftmargin=*}

\usepackage{xspace}
\newcommand{\gpm}{{graph pattern mining}\xspace}
\newcommand{\GPM}{{Graph Pattern Mining}\xspace}
\newcommand{\hrate}{{hit rate}\xspace}

\newcommand{\dg}{{$\mathcal{G}$}\xspace}
\newcommand{\pg}{{$\mathcal{P}$}\xspace}

\newcommand{\sys}{{\textproc{ScaleGPM}}\xspace}

\newcommand{\NS}{{\texttt{NS}}\xspace}
\newcommand{\GS}{{\texttt{GS}}\xspace}
\newcommand{\lazy}{{{lazy-verify}}\xspace}
\newcommand{\eager}{{{eager-verify}}\xspace}
\newcommand{\hybrid}{{hybrid sampling}\xspace}
\newcommand{\Eager}{{{Eager Verify}}\xspace}

\newcommand{\usoveraryaall}{{{$565\times$}}\xspace}
\newcommand{\usoveraryaallmax}{{{$610169\times $}}\xspace}

\newcommand{\usoveraryamotif}{{{$125 \times$}}\xspace}
\newcommand{\usoverexactmotif}{{{$10563 \times$}}\xspace}

\newcommand{\usoveraryalargegraph}{{{$130 \times$}}\xspace}
\newcommand{\usoverexactlargegraph}{{{$7245 \times$}}\xspace}

\newcommand{\usoveraryalargepattern}{{{$599 \times$}}\xspace}
\newcommand{\usoverexactlargepattern}{{{$27641 \times $}}\xspace}

\newcommand{\usoveraryaclique}{{{$2747 \times $}}\xspace}
\newcommand{\usoveraryacliquemax}{{{$610169 \times$}}\xspace}

\newcommand{\usoverexactclique}{{{$4045 \times$}}\xspace}
\newcommand{\usoverexactcliquemax}{{{$65525 \times $}}\xspace}

\newcommand{\gsoverns}{{{$13 \times$}}\xspace}
\newcommand{\gsovernsmax}{{{$61 \times$}}\xspace}

\newcommand{\probP}{\text{I\kern-0.15em P}}

\crefformat{section}{\S#2#1#3}

\newcommand{\hl}[1]{\textcolor{red}{#1}}

\newcommand{\hlg}[1]{\textcolor{codegreen}{#1}}

\newtheorem{theorem}{Theorem}

\def\BibTeX{{\rm B\kern-.05em{\sc i\kern-.025em b}\kern-.08em
    T\kern-.1667em\lower.7ex\hbox{E}\kern-.125emX}}

\begin{document}

\title{Accurate and Fast Approximate \GPM at Scale}

\author{Anna Arpaci-Dusseau, Zixiang Zhou, Xuhao Chen}
\affiliation{      
    \institution{Massachusetts Institute of Technology}
    \streetaddress{32 Vassar St}
    \city{Cambridge} 
    \state{MA} 
    \postcode{02139}
    \country{USA}
}
\email{{annaad, zpeter, cxh}@mit.edu}
\date{}

\begin{abstract}
   Approximate \gpm (A-GPM) is an important data analysis tool for numerous graph-based applications.
   There exist sampling-based A-GPM systems to provide automation and generalization over a wide variety of use cases. 
   Despite improved usability, there are two major obstacles that prevent existing A-GPM systems being adopted in practice. 
   First, the termination mechanism that decides when to terminate sampling lacks theoretical backup on confidence, and performs significantly unstable and thus slow in practice.
   Second, they particularly suffer poor performance when dealing with the ``needle-in-the-hay'' cases, because a huge number of samples are required to converge, given the extremely low \hrate of their lazy-pruning strategy and fixed sampling schemes.
   
   We build \sys, an accurate and fast A-GPM system that removes the two obstacles.
   First, we propose a novel on-the-fly convergence detection mechanism to achieve stable termination and provide theoretical guarantee on the confidence, with negligible online overhead.
   Second, we propose two techniques to deal with the ``needle-in-the-hay'' problem, \textit{\eager} and \textit{\hybrid}. 
   Our \eager method drastically improves sampling \hrate by pruning unpromising candidates as early as possible. 
   Hybrid sampling further improves performance by automatically choosing the better scheme between fine-grained and coarse-grained sampling schemes. 
   Experiments show that our online convergence detection mechanism can precisely detect convergence, and results in stable and rapid termination with theoretically guaranteed confidence.
   We also show the effectiveness of \eager in improving the \hrate,
   and the scheme-selection mechanism in correctly choosing the better scheme for various cases.
   Overall, \sys achieves an \texttt{geomean} average of \usoveraryaall (up to \usoveraryaallmax) 
   speedup over the state-of-the-art A-GPM system, Arya.
   In particular, \sys handles billion-scale graphs in seconds, where existing systems either run out of memory or fail to complete in hours.
\end{abstract}

\maketitle

\section{Introduction}

\GPM (GPM)~\cite{Arabesque,RStream,Peregrine,AutoMine,Fractal,GraphPi,Sandslash,G2Miner}, which searches for instances of small patterns (e.g., triangles) in a data graph, is a key building block in graph-based data mining.
Despite its wide adoption in real-world applications, GPM is hard to scale to large problem size as it is computationally quite expensive~\cite{G2Miner}.
Fortunately, many real-world use cases do not require exact GPM solutions.
For instance, to characterize network structures in biochemistry, ecology and engineering~\cite{Motifs2}, we need only to estimate the pattern (a.k.a motif) frequency distribution, instead of exactly counting motifs.
Therefore we only need Approximate GPM (A-GPM),
which trades accuracy for speed, as long as some error bound is met in a given application context.

There have been A-GPM systems, e.g., ASAP~\cite{ASAP} and Arya~\cite{Arya}, that are built to simplify programming and improve usability.
The system responsibility is two fold. 
First, they provide \textit{generalized} sampling techniques for arbitrary patterns and programming APIs to support a wide variety of use cases.
Second, they provide automated mechanisms to determine key sampling parameters, e.g., how many samples required to draw for a given error bound.

However, existing A-GPM systems have two major limitations that prevent them from being adopted in practice. 
First, the automated termination mechanism to determine when it is confident enough to terminate sampling, does not have strong theoretical backup on confidence, and thus is significantly unstable and slow in practice (see \cref{subsect:unstable-termination} for detail).
In existing systems, sampling is terminated when the predicted number of samples $N_s$ have been drawn.
They predict $N_s$ based on an offline error-latency profiling (ELP) procedure before launching the sampling. 
However, the ELP prediction of $N_s$ is dependent on the true count, which is supposed to be estimated on the $N_s$ samples, creating a circular dependency.
Therefore, the ELP cannot theoretically guarantee confidence in bounding the error.
Also, the ELP returns \emph{unstable} predictions, i.e., huge variances or even failures in predicting $N_s$, which leads to poor performance.

The second limitation is the poor performance when dealing with the extremely sparse, a.k.a, \textit{needle-in-the-hay} cases, where only a small number of matches of the pattern exist in a big graph (see \cref{fig:needle}).
In this case, neighbor sampling (\NS), the sampling scheme used in ASAP and Arya, fails to hit a match in most of the samples, leading to a very low sampling \hrate, e.g., 0.0000017\%. 
This results in slow convergence as a huge number of samples are required to meet the error bound.
In some extreme cases, we observe that Arya can be even slower than exact GPM solutions.

To address the unstable termination problem, 
we propose a novel on-the-fly convergence detection method for \NS. 
Instead of offline predicting the termination condition before execution, our \textit{online} method dynamically collects statistics and predicts errors during the sampling execution, until convergence.
The convergence is detected when the predicted error is below the user specified error bound. 
We prove formally in \cref{thm:error} that the probability of the true error being smaller than our predicted error is $1-\delta$, for a given confidence $1-\delta$.
This provides us theoretical guarantee in confidence, which prior systems lack.
Meanwhile, our detected termination points are stable across different runs while the ELP method in prior systems is highly unstable.
Therefore, our online method can significantly accelerate the execution, because it needs much fewer samples than prior systems, while causing negligible online overhead as statistics can be trivially collected.

As for the low \hrate problem in \NS, 
we find that the root cause is the delayed check on pattern closure in prior systems, which we call \lazy.
We then introduce \textit{\eager} that applying pruning at its earliest possible point to avoid unpromising candidates and thus improve hit rates.
We prove theoretically that \eager is unbiased,
and show that it drastically improves the sampling \hrate, 
which in turn significantly accelerate convergence.
Furthermore, for extremely sparse cases where the fine-grained \NS sampling scheme can not handle, 
we propose a \textit{\hybrid} method by adaptively selecting the better sampling scheme, between \NS and a coarse-grained scheme, graph sparsification (\GS), for various graphs and patterns.
For scheme selection, we build a cost model for each sampling scheme to estimate the execution time.
With our models, \hybrid manages to always select the cheaper one from \NS and \GS for a wide variety of test cases. 

We then build \sys, an accurate and fast A-GPM system
that incorporates our proposed mechanisms: online convergence detection, \eager and \hybrid.
\sys efficiently leverages parallel hardware and provides flexible modes to meet various accuracy requirements with confidence.
Our experiments on a multicore CPU show that, 
with orders-of-magnitude reduction in required samples and improvement in hit rates,
\sys achieves an average of \usoveraryaall (up to \usoveraryaallmax) speedup over the state-of-the-art A-GPM system, Arya.
The hybrid method further improves performance by \gsovernsmax on extremely hard (i.e. sparse) cases.
Particularly, \sys handles billion-scale graphs in seconds where previous frameworks either
run out of memory or fail to complete in hours.
This paper makes the following contributions:
\begin{itemize}
\item We conduct analysis and empirical study on sampling schemes, and identify the two major limitations and their root causes in existing A-GPM systems.
\item We propose a novel on-the-fly convergence detection method for the \NS sampling scheme, which is the first to provide theoretical guarantee on confidence.
\item We introduce the \eager mechanism to improve \hrate of the \NS scheme and thus achieve faster convergence speed. 
\item We propose hybrid sampling to further improve performance, by adaptively selecting a better scheme based on cost models.
\item We build \sys that incorporates the above novel mechanisms,
and evaluation shows that it significantly outperforms the state-of-the-art system, and efficiently handles huge graphs.
\end{itemize}

\section{Background}
\vspace{-5pt}
\subsection{Approximate \GPM} \label{subsect:agpm}

\GPM (GPM) finds subgraphs that match 
given pattern(s) \pg in a given {data graph} \dg. 
There exist \textit{explicit} GPM tasks like subgraph counting (SC) and motif counting (MC), and \textit{implicit} tasks like frequent subgraph mining (FSM)~\cite{Pangolin}.
GPM has numerous applications in {AI} and {big-data}~\cite{G2Miner}, 
including bioinformatics, chemical engineering, fraud detection, social network analysis, recommender systems, etc. 

Exact GPM is solved by enumerating subgraphs in the data graph and searching for matches.
The search space can be defined as a \textit{subgraph tree}:
each node in the tree is a subgraph of the data graph \dg. 
A subgraph $H$ in level $l$ of the tree have $l$ vertices. The root (level 0) is an empty subgraph. 
A parent node $H$ at level $i$ can be \textit{extended} to a child node $H'$, 
by adding a vertex/edge in $H$'s neighborhood in \dg, i.e., $H'=H+\{v\}$, $v \in \mathcal{N}(H)$.
Each leaf of the tree is a candidate of match,
which is then compared with the pattern \pg, to test if it is a match.
Pruning schemes, e.g., {matching order}~\cite{G2Miner,Peregrine,AutoMine}, {symmetry breaking}~\cite{GraphZero,Pangolin,Sandslash} and decomposition~\cite{ESCAPE,DecoMine}, 
are applied to reduce the search space (i.e., prune the subgraph tree).
Nevertheless, this search is extremely expensive, as the computational complexity increases exponentially in the size of the pattern.

Many real-world use cases do not require exact GPM solutions.
For example, when we use motif (a.k.a graphlet) distribution as a ``signature'' (e.g., graph similarity) for social network analysis or fraud detection, 
it is quite sufficient to just provide approximate counts of the motifs.
Also in FSM the users only want to find those \textit{frequent} patterns whose occurrences are above some user specified threshold, or simply top-$K$ most frequent patterns,
where estimated counts would be sufficient to give high quality solutions.
Therefore, for all these use cases, we can perform Approximate GPM (A-GPM) instead to substantially reduce the total amount of computation.

In this paper we focus on sampling based A-GPM approaches. 
Generally, such an approach first samples a portion of the graph, searches for match in the sample,
and makes the estimation by scaling the sampled result.
This process can be repeated multiple times to improve the confidence of estimation.
Formally, given \dg and \pg, an A-GPM solver aims to use a randomized ($\epsilon$, $\delta$)-approximation scheme, 
which estimates the number of (non-induced or induced) occurrences of \pg in \dg  
within a factor of ($1 \pm \epsilon$), 
with probability at least $1 - \delta$, 
where $\epsilon$ and $\delta$ are user defined parameters.
There are a large volume of studies on A-GPM applications, such as triangle counting~\cite{Buriol,ColorfulTC,SpectralTC,DOULION,NeighborhoodSampling,RevisitTC,TETRIS,TC-FourCycle,SampleHold}, clique/cycle counting~\cite{LightningKclique,Five-Cycle,TC-FourCycle}, motif counting~\cite{BeyondTri,PathSampling,fascia,GUISE,Bressan,Motivo,Bressan1,Lifting,PSRW,ParSE,CC-biomolecular}, butterfly counting~\cite{ButterflyCounting}, frequent subgraph mining~\cite{GREW,MaNIACS,Scalemine,Ap-FSM,ASGS}. They all use sampling to reduce computation, though their sampling schemes are \textit{customized} for the specific problems.
These custom implementations do not offer system support, like automated termination, generic APIs, or choices in speed and error trade-off.

\subsection{Sampling Schemes for GPM Problems}

\begin{algorithm}[tb]
\footnotesize
\caption{Neighbor Sampling}
\label{algo:ns}
\begin{algorithmic}[1]




\For{{\bf each} sampler $i \in$ [1, $N_s$]} \label{algo:ns:line:each-sampler}

\State $C_i \gets 0, \alpha \gets m$  
\State Sample an edge $e_1$ from $\mathcal{E}$, let $E_1 \gets \{e_1\}$ \label{algo:ns:line:sample-first-edge}

\For{$j$ in [2, $k-1$]}
\State Sample a neighboring edge $e_j$ from $E_{j-1}$'s neighborhood \label{algo:ns:line:sample-neighbor} 
\State $E_j \gets E_{j-1} \bigcup \{e_j\}$,
 $\alpha \gets \alpha \times c_j$ \Comment{\hlg{$c_j$ is the neighborhood size}} \label{algo:ns:line:scaling}
\EndFor


\If{closing edges for ($E_{k-1}$, \pg) exist in \dg} 
 $C_i \gets \alpha$  \Comment{\hlg{closure check}} \label{algo:ns:line:closure}
\EndIf
\EndFor

\State $C' = \sum C_i / N_s$ is the estimated count  \label{algo:ns:line:average}

\end{algorithmic}
\end{algorithm}

\subsubsection{Neighbor Sampling (\NS)}

\cref{algo:ns} shows how neighbor sampling works.
For each sample,
it starts with sampling one edge from \dg uniformly at random (\cref{algo:ns:line:sample-first-edge}), 
and then repetitively samples one more edge from the neighborhood of the currently sampled edges (\cref{algo:ns:line:sample-neighbor}), 
until the size (number of sampled vertices) is the same as the pattern.
It then does \textit{closure check} (\cref{algo:ns:line:closure}), i.e., check in \dg the existence of the {\it closing edges}, which form a match of the pattern together with existing edges.  
We can draw multiple samples (\cref{algo:ns:line:each-sampler}) in parallel, and average them for improved accuracy (\cref{algo:ns:line:average}). 

For a given match $\hat{M}$ where $\hat{e_1}$ is the first edge, $Pr[\hat{M}]=\frac{1}{m\cdot \prod_{2}^{k-1} c_j}$, where $c_j=|\mathcal{N}(H_{j-1})|$, $H_{j-1}$ is the $E_{j-1}$ induced subgraph, and $\mathcal{N}(H)$ is the neighbor set of $H$.
So the count is scaled by $\alpha=m\cdot \prod_{2}^{k-1} c_j$ (\cref{algo:ns:line:scaling}).
Apparently $\alpha$ varies in different samples, meaning matches are not equally likely to be sampled.
Thus 
$\alpha$ is maintained for each sample and used for normalization.

\NS has been widely used in A-GPM applications~\cite{NeighborhoodSampling,Near-Optimal,Lifting,RevisitTC} with customized optimizations.
In NS, each time an arbitrary neighbor is sampled from $H_i$'s neighborhood, which may lead to a high failure rate in the closure check.
Therefore, restrictions are added to make the sampled subgraph more likely be a match. For example, if \pg is a cycle, we can sample a path~\cite{PathSampling}, i.e., restrict the next neighbor to be from the two endpoints' neighborhood.
For an arbitrary pattern \pg, we can do neighbor sampling following \pg's spanning tree~\cite{Bressan}. 
Furthermore, automorphism check can be added to avoid redundant subgraphs~\cite{SampleMine}. 
More generally, we can sample multiple neighbors each time, instead of a single neighbor~\cite{SampleMine}. 



\begin{algorithm}[t]
\footnotesize
\caption{Edge Sparsification}
\label{algo:espar}
\begin{algorithmic}[1]




\State Randomly select a subset $E'$ of $p\times m$ edges from $E$, $m=|E|$, $0 < p \leq 1$ \label{algo:espar:line:sample}
\State Generate an induced subgraph $G'$ = ($V$, $E'$)  \label{algo:espar:line:subgen}
\State $C'$ += \Call{ExactCounting}{$\mathcal{G}'$, \pg} 
\Comment{\hlg{Exact counting on $\mathcal{G}'$}} 
\label{algo:espar:line:exact-count}

\State $C = C' \times p^{-l}$ is the estimated count in $\mathcal{G}$ \Comment{\hlg{$l$ is \# of edges in \pg}}

\end{algorithmic}
\end{algorithm}

\subsubsection{Subgraph Sampling} \label{subsect:subgraph-sampling}

The idea of subgraph sampling is to sample a subset of vertices or edges from \dg with some probability $p$, to form a subgraph $G'$, and then do exact search in $G'$, and finally scale the count based on $p$.
One popular way to sample a subgraph is to use the {\it graph sparsification} (\GS) technique~\cite{ApproxMinCut,GFforGraphSpar,gpar-gspar,gspar-resis,gspar-spectral},
which sparsifies \dg by randomly removing some edges.
Bernoulli Edge Sparsification (BES)~\cite{DOULION} is such an example, as shown in \cref{algo:espar}.
For each edge in \dg, include it with probability $p$ (\cref{algo:espar:line:sample}) to get graph $G'$ (\cref{algo:espar:line:subgen}).
For each match $M$ of \pg in \dg, the probability that $M$ exists in $G'$ is $p^l$, hence the expected count in $G'$ is $C' = p^l\times C$.
BES is simple and easy to implement, and can be trivially parallelized.

Color Sparsification~\cite{ColorfulTC} (CS), another \GS approach, sparsifies \dg by first randomly assigning a color from \{1, 2, $\cdots$ , $c$\} to each vertex and then only preserving edges whose two endpoints are in the same color. The probability of choosing an edge is $p= 1/c$, and the probability of choosing a match is $p^{l-1}$.
Hence the expected count in $G'$ is $C' = p^{l-1}\times C$.
Apparently, the probability of preserving a match of \pg is higher than that in BES. In other words, it could meet the same error bound with more edges removed, thus less computation.
The downside is that different edges are not independent of each other any more, which means potentially higher variance. 
Usually a single sampler is used for sparsification, though more samplers can be used to improve confidence.
 

Another way to sample a subgraph is Egonet Sampling~\cite{EgonetMotif,ButterflyCounting}, in which an element (vertex or edge) of \dg is sampled and the \textit{egonet}, i.e., the local neighborhood, of this element is extracted as a subgraph.
It often requires many samplers.

One advantage of subgraph sampling is that a state-of-the-art exact counting algorithm can be applied on the sampled subgraph. 
But the cost is extra time and space to extract the subgraph(s), and the time complexity is still exponential in the size of pattern $k$.
Subgraph sampling has only been used in A-GPM applications for specific patterns~\cite{DOULION,ColorfulTC,SpectralTC,ParallelCC,ButterflyCounting,Five-Cycle,EgonetMotif}.


\subsection{Approximate GPM Systems}

A-GPM systems, such as ASAP~\cite{ASAP} and Arya~\cite{Arya}, have been proposed to simplify A-GPM programming.
These systems provide APIs to the users for them to easily compose various A-GPM applications.
As opposed to those case-by-case customized implementations, an A-GPM system provides a {generalized}, sampling-based approximation method, for arbitrary patterns.
Moreover, instead of hand tuning the key sampling parameters, e.g., the number of samples, in hand-implemented applications,
these systems provide the Error-Latency Profile (ELP) method to automatically choose the values of the sampling parameters for each specific case, e.g., input data graph and pattern, to meet the user specified error bound.

Nevertheless, all the prior systems use fixed sampling schemes to make approximation. 
For example, ASAP uses the \NS scheme, and implements \NS in the \textit{edge streaming} fashion 
~\cite{StreamingTriangleCounting,StreamingCount,StarCounting,EdgeSampling,heavyEdge} where edges are streamed in as a sequence, instead of loaded all at once, to save memory space.
Arya is also based on \NS, but adds pattern decomposition on top of ASAP, 
to reduce the amount of work in each sample for large, easy-to-decompose patterns.
Since both ASAP and Arya are based on \NS, they both suffer the shortcomings of \NS,
which are discussed in detail next.





\section{Understanding Sampling Tradeoffs}
\label{sect:motiv}


\subsection{Termination Condition and Confidence} \label{subsect:unstable-termination}

\begin{figure}[tb]
\centering
\includegraphics[width=0.49\textwidth]{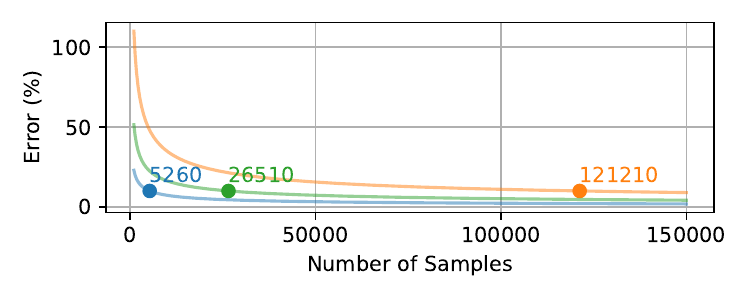}
\vspace{-11pt}
\caption{Three different runs (three curves) of Arya's ELP prediction given the \texttt{LiveJ} graph and \texttt{triangle} pattern. With an error bound of 10\%, the curves give dramatically different prediction on the number of samples $N_s$: 5,260, 26,510 and 121,210. This leads to a 25$\times$ performance difference in the sampling execution phase.}
\label{fig:elp-unstable}
\end{figure}

A major responsibility of an A-GPM system is to decide when the sampling can be terminated with enough confidence to meet the error bound. 
Existing A-GPM systems use error-latency profile (ELP), before the execution of the sampling procedure, to pre-determine the number of samples $N_s$ required.
The execution is then terminated simply when $N_s$ samples have been drawn.
However, we observe that the value of $N_s$ that ELP predicts vary dramatically across different runs of ELP.
\cref{fig:elp-unstable} shows the results of three runs of ELP.
Each curve represents the prediction of one ELP run.
We get predictions of $N_s$ as 5,260, 26,510 and 121,210, respectively.
This huge prediction difference leads to a 25$\times$ performance difference in the sampling execution phase.

More importantly, the ELP method for termination condition adds an indirection between the error estimation and the confidence.
Originally, the required number of samples $N_s$ is derived from the Chernoff bound in ASAP~\cite{ASAP} or Chebyshev’s inequality in Arya~\cite{Arya}.
For example, given $\epsilon$ and $\delta$, the lower bound in Arya is $N_s \geq \frac{K\cdot m\cdot \rho}{C\cdot \epsilon^2\delta}$, where $K$ is a constant, $m=|E|$, $\rho$ is a pattern specific constant, and $C$ is the true count.
The problem is, this bound contains the true count $C$, which is the output that is supposed to be estimated.
Therefore, what ELP does is to essentially first estimate $C$ without theoretical confidence-error ($\delta$-$\epsilon$) guarantees,
then feed it into the lower bound to get $N_s$.
$N_s$ is then used to do sampling and estimate a more accurate $C$. Note that ELP estimates $C$ by sampling in a sparsified graph and iteratively updating the parameters and the estimates until convergence.
Although the bound inequation used to calculate $N_s$ contains $\delta$,
estimating $C$ by ELP \textit{has no involvement of $\delta$}.
It means the estimation of $N_s$ loses connection to the confidence.
The root cause is the \emph{circular dependency} between $C$ and $N_s$, i.e., $C$ is used to estimate $N_s$, and $N_s$ is used to estimate $C$, which is fundamentally unavoidable in the ELP based approach. 

Due to the \emph{circular dependency}, ELP provides only a heuristic rather than a strong theoretical bound.
Another limitation of ELP is its own convergence speed. We observe that in some cases, ELP fails to converge within 10 hours (see details in \cref{tab:clique}).

\begin{figure}[tb]
\centering
\includegraphics[width=0.45\textwidth]{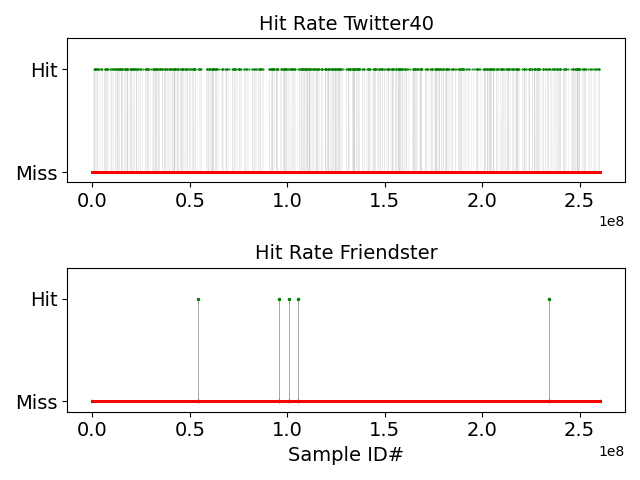}
\caption{Sample hits and misses in Arya,  
on \texttt{Twitter40} (top) and \texttt{Friendster} (bottom) graphs.
The pattern is \texttt{4-clique} for both.
In total $10^8$ samples are drawn in both cases. 
Each green point is a hit sample, while each red point is a miss sample.
For \texttt{Twitter40}, there are 7,033 hits with a hit rate of $7\times 10^{-5}$.
For \texttt{Friendster}, there are only 5 hits with a $5\times 10^{-8}$ hit rate (i.e. needle in the hay).
}
\label{fig:needle}
\end{figure}
\vspace{-3pt} 

\subsection{Characterizing Neighbor Sampling}

For the sampling based approximation approach, the estimation difficulty depends on the density and distribution of the matches of \pg in the graph \dg.
When there are plenty of matches, defined as the \textit{dense case}, it is easier to make estimation than the \textit{sparse case}, where there exist only a few matches.
This is because it is more likely to draw a \textit{successful} sample (i.e., find a match) if there are more matches,
and the confidence to meet an error bound depends on seeing enough number of successful samples.
Note that the extremely sparse case is known as finding the \textit{needle in the hay}. 

We call a sample that finds a match a \textit{hit}, otherwise a \textit{miss}.
\cref{fig:needle} shows how the hits and misses are distributed in $1\times10^8$ samples drawn by Arya~\cite{Arya}, for the \texttt{4-clique} pattern on graphs \texttt{Twitter} (top) and \texttt{Friendster} (bottom). 
The two graphs have similar sizes but quite different degree distribution (see maximum degree in \cref{tab:input}), and thus different hit rates.
There are 7,033 hits in \texttt{Twitter} ($7\times10^{-5}$ hit rate), while only 5 hits appeared in \texttt{Friendster} ($5\times10^{-8}$ hit rate) which is a typical needle-in-the-hay case. 
Since the execution time is roughly linear in $N_s$, this difference in hit rates would result in a $\sim$700$\times$ execution time difference in practice, assuming the same average time per sample.

It is known that \NS works poorly in the sparse case, for example, when \pg is dense and \dg is sparse, and particularly in the case of needle in the hay~\cite{Arya}.
One can expect that in this case, the closure check fails frequently, which means a large number of samples $N_s$ is required to meet a certain error bound.
Moreover, dense patterns are particularly problematic for Arya~\cite{Arya} which decomposes the pattern into sub-patterns, as a dense pattern would be split with many edgecuts, and checking closure is quite expensive.
Arya thus suffers significant slowdown in those cases, and in some extreme cases it is even slower than the exact solution.

\begin{figure}[tb]
\centering

\begin{subfigure}[]{0.23\textwidth}
\includegraphics[width=0.9\textwidth]{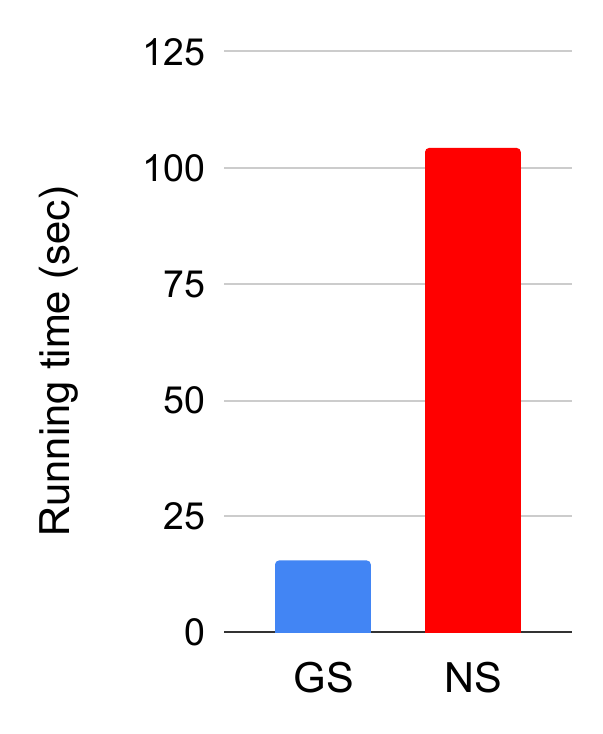}
\vspace{-10pt}
\caption{{\texttt{house} on \texttt{fr} with 0.1\% error}}
\label{fig:fs-house}
\end{subfigure}
\hfill
\begin{subfigure}[]{0.23\textwidth}
\includegraphics[width=0.78\textwidth]{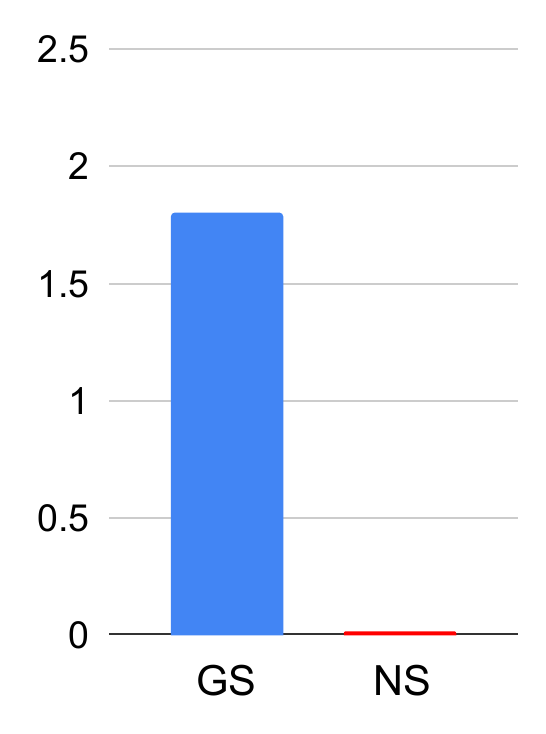}
\vspace{-10pt}
\caption{{\texttt{5path} on \texttt{livej} with 2\% error}}
\label{fig:mico-5path}
\end{subfigure}


\caption{Execution time variance of Neighbor Sampling (\NS) and Graph Sparsification (\GS), under the same error bounds. }
\label{fig:motivation}
\end{figure}

\subsection{Coarse-grain vs. Fine-grain Sampling}

A key difference between \NS and \GS is the granularity of each sample.
Since each sample contains at most one match, we classify \NS to be a \textit{fine-grain} sampling scheme,
as opposite to \textit{coarse-grain} sampling schemes, each of whose samples can contain multiple matches.
For \NS, each sample is in the size of the pattern size $k$, and each sample task is lightweight. But we need a lot of samples, i.e., $N_s \gg 1$, to get a meaningful estimation, since each sample contains at most one match. 
In contrast, subgraph sampling schemes, including \GS, are coarse grained.
In \GS, a sample is a sparsified graph, which could potentially contain many matches in it. 
As the sample granularity in \GS is much larger, it performs better than \NS when given a sparse case, as it is more likely to hit matches in a big region of the graph. However, this advantage comes at the cost of several drawbacks. First, each sample in \GS is a much larger computational task, and the complexity is exponential in the pattern size. Second, as multiple matches appear in the same sample, \GS may yield worse variance than \NS in the worst case.




To summarize, given the distinct characteristics in the input data (graph and pattern), 
none of these sampling schemes can always be the best solution.
\cref{fig:motivation} compares the running time of \NS with \GS, when given the same error bound. 
On the left we mine the \texttt{house} pattern on the \texttt{Friendster} graph with an error bound of 0.1\%, 
where \GS is 6.7$\times$ \textbf{faster} than \NS.
On the right we mine the \texttt{5-path} pattern in the \texttt{Livej} graph with an error bound of 2\%, 
where \GS is 10$\times$ \textbf{slower} than \NS.

\section{Proposed Mechanisms} \label{sect:main-ideas}

To achieve stable termination with confidence, we propose a novel on-the-fly convergence detection mechanism in \cref{subsect:converge} that is fundamentally different from the existing ELP approach. 
To improve \hrate in \NS, we introduce \eager and prove it unbiased in \cref{subsect:ns-prune}.
To further improve performance in handling needle-in-the-hay cases, 
we propose a hybrid method that adaptively selects the best-performing sampling scheme.
For scheme selection, we establish cost models (a.k.a. performance models) for the \NS (\cref{subsect:ns-perf-model}) and \GS (\cref{subsect:gs-perf-model}) scheme to estimate their execution time, and select the faster one.
We only focus on the two schemes, but this hybrid method can be extended to include other schemes in the future.

\subsection{Online Convergence Detection} \label{subsect:converge}




Due to the circular dependency issue (\cref{subsect:unstable-termination}), the ELP method breaks the theoretical guarantee on confidence.
Also, it is fundamentally hard for ELP to establish confidence because it is done before execution, and there is little information we can leverage.
Therefore, we propose an on-the-fly approach to establish confidence.
This is based on our observation in \NS sampling procedures, that the estimation errors tend to converge over time (see \cref{fig:ns-base-ns-prune}).
Instead of predetermining the required number of samples offline (i.e. before execution),
we detect online if the estimates have converged, and then terminate the execution subsequently.
However, convergence detection is not trivial. 
A straightforward method is to check if the difference of the error curve is small enough, i.e., within a fixed threshold. 
But this does not work because the termination point depends on the user defined confidence and error bound.
The key challenge is then how to decide termination to meet error bound with confidence.

To address it,
we propose to {predict} the error online periodically with confidence, and terminate when the predicted error is below the error bound.
To make predictions with confidence, we collect online statistics during execution,
which allows us to formally {derive} predicted errors based on the probability theory.
Our key insight is that confidence can be established by the normal distribution of sampled means (formally proved in \cref{thm:error}), as shown in \cref{fig:ns-normal}.
Specifically, we keep track of the \textit{mean} $\mu$ of the estimated counts and the \textit{standard deviation} $\sigma$ of the means, 
and then use $\mu$, $\sigma$ and $\delta$ to compute a relative error $\hat{\epsilon}$ using \cref{equ:error}, where $\Phi^{-1}$ is inverse of the cumulative distribution function of the standard normal.
Note that $\mu$ is also used as the estimate of the true count $C$.

\vspace{-5pt}
\begin{equation} \label{equ:error}
\hat{\epsilon} = \frac{\Phi^{-1}(1-\frac{\delta}{2})\sigma}{\mu}
\end{equation}
\vspace{-5pt}

When we detect that $\hat{\epsilon}$ is below the user specified error bound $\epsilon$, according to \cref{thm:error}, we can safely terminate the sampling, and conclude that $\mu$ is an estimate of $C$, under an error bound $\epsilon$ with a confidence of $1-\delta$.
We refer our approach as \NS-online.

\begin{theorem} \label{thm:error}
    Given $\delta$, $n$ samples $X_1,\dots,X_n$ drawn by using the \NS sampling scheme, and the mean of sampled counts $\mu = \frac{1}{n}\sum_{i=1}^n X_i$, let $C$ be the true count and $\hat{\epsilon}$ be the estimated error computed by \cref{equ:error}. As $n \to \infty$, the probability of the true relative error being smaller than the estimated relative error is $\probP\left(\frac{|\mu-C|}{C} < \hat{\epsilon}\right) = 1-\delta$.
\end{theorem}

\begin{figure}[tb]
\includegraphics[width=0.4\textwidth]{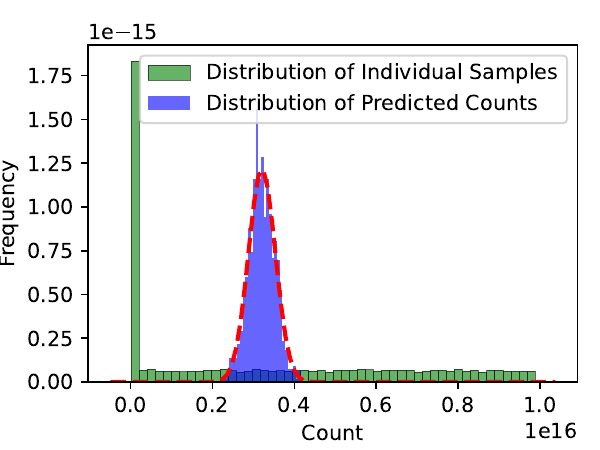}
\caption{ The normal distribution of the means of sampled counts (i.e. our predicted counts) using neighbor sampling (\NS). We ran \NS to collect $10^6$ samples on \texttt{LiveJ}, 4-\texttt{clique}. We obtained a predicted count by taking the mean of a random subset of 100 of these underlying samples. 
We simulated 1000 of these predicted counts. 
Although the underlying distribution of the sampled counts (green bars) is not a normal distribution, their means (purple bars), which are our predicted counts, do follow a normal distribution (dashed red line).
}
\label{fig:ns-normal}
\end{figure}

\begin{proof}[Proof]
Let $D$ be the distribution $X_1,\dots,X_n$ are sampled from. Since the \NS estimator is unbiased,
the true count $C$ is the mean of the distribution $D$.
Our estimator of $C$ is $\mu := \frac{1}{n}\sum_{i=1}^n X_i$. This estimator satisfies $\mathbb{E}[\mu] = \mathbb{E}[\frac{1}{n}\sum_{i=1}^n X_i] = \frac{1}{n}\sum_{i=1}^n \mathbb{E}[X_i] = \frac{1}{n}\sum_{i=1}^n C = C$ and $\mathbb{V}ar[\mu] = \mathbb{V}ar[\frac{1}{n}\sum_{i=1}^n X_i] = \frac{1}{n^2}\sum_{i=1}^n \mathbb{V}ar[X_i] = \frac{1}{n}\mathbb{V}ar[X_1]$. Because $\mu$ is the average of $n$ samples from a distribution $D$ (which clearly has finite mean and variance), the central limit theorem (CLT) applies, so $\mu$ follows the normal distribution.
Formally, $\frac{\mu-\mathbb{E}[\mu]}{\sqrt{\mathbb{V}ar[\mu]}}$ converges in distribution to a standard normal as $n \to \infty$. 

By the law of large numbers, as $n \to \infty$, the sample variance $\frac{1}{n}\sum_{i=1}^n X_i^2 - \mu^2$ converges in probability to $\mathbb{V}ar[X_1]$. Letting $\sigma^2 := \frac{1}{n}\left(\frac{1}{n}\sum_{i=1}^n X_i^2 - \mu^2\right)$, this means that $\frac{\mathbb{V}ar[\mu]}{\sigma^2} = \frac{\frac{1}{n}\mathbb{V}ar[X_1]}{\frac{1}{n}\left(\frac{1}{n}\sum_{i=1}^n X_i^2 - \mu^2\right)} = \frac{\mathbb{V}ar[X_1]}{\left(\frac{1}{n}\sum_{i=1}^n X_i^2 - \mu^2\right)}$ converges in probability to $1$. Additionally, $\mu$ converges in probability to $C$, so $\frac{\mu}{C}$ converges in probability to $1$. Therefore, by Slutsky's Theorem, $\frac{\mu-\mathbb{E}[\mu]}{\sqrt{\mathbb{V}ar[\mu]}} \cdot \sqrt{\frac{\mathbb{V}ar[\mu]}{\sigma^2}} \cdot \frac{\mu}{C} = \frac{\mu-C}{C} \cdot \frac{\mu}{\sigma}$ converges in distribution to a standard normal. This means that for any fixed $x$, we have
\[
\probP\left(\frac{\mu-C}{C} \cdot \frac{\mu}{\sigma} < x\right) \to \Phi(x)\text{ as }n \to \infty.
\]
This implies that
\[
\probP\left(\left|\frac{\mu-C}{C} \cdot \frac{\mu}{\sigma}\right| < x\right) \to 2\Phi(x)-1\text{ as }n \to \infty.
\]
Plugging in $x = \Phi^{-1}(1-\frac{\delta}{2})$, we get
\[
\probP\left(\left|\frac{\mu-C}{C} \cdot \frac{\mu}{\sigma}\right| < \Phi^{-1}(1-\frac{\delta}{2})\right) \to 2\Phi\left(\Phi^{-1}(1-\frac{\delta}{2})\right)-1\text{ as }n \to \infty.
\]
Rearranging terms and simplifying yields
\[
\probP\left(\left|\frac{\mu-C}{C}\right| < \frac{\Phi^{-1}(1-\frac{\delta}{2})\sigma}{\mu}\right) \to 2\left(1-\frac{\delta}{2}\right)-1\text{ as }n \to \infty
\]
\[
\implies \probP\left(\left|\frac{\mu-C}{C}\right| < \hat{\epsilon}\right) \to 1-\delta\text{ as }n \to \infty.
\]
\end{proof}

\algdef{SE}[DOWHILE]{Do}{doWhile}{\algorithmicdo}[1]{\algorithmicwhile\ #1}%

\begin{algorithm}[t]
\footnotesize
\caption{\NS-online convergence detection}
\label{algo:ns-online}
\begin{algorithmic}[1]

\State sum $\gets 0$, squaredSum $\gets 0$, $n=0$, $W \gets N_{min}$

\While{$\hat{\epsilon}$ > $\epsilon$} \label{algo:ns-online:line:comp_error}

\For{{\bf each} sampler $i \in$ [1, $W$]}  \Comment{$W$ is the window size}  \label{algo:ns-online:line:period}

\State $X_i \gets \Call{DrawASample}{ }$  \Comment{$X_i$ is the $i$-th sampled count}  \label{algo:ns-online:line:each-sampler}

\State sum $\gets$ sum $+ X_i$   \Comment{$\sum_{i=1}^n X_i$} \label{algo:ns-online:line:sum}

\State squaredSum $\gets$ squaredSum $ + X_i * X_i$  \Comment{$\sum_{i=1}^n X_i^2$} \label{algo:ns-online:line:squaredSum}

\State $n \gets n + 1$

\EndFor

\State $\mu \gets $ sum / $n$ \Comment{mean of sampled counts}

\State var $\gets$ squaredSum / $n - \mu * \mu$ \Comment{variance of sampled counts}

\State $\sigma \gets$ \texttt{sqrt}(var/$n$)   \Comment{standard deviation} \label{algo:ns-online:line:std_dev}

\State $\hat{\epsilon} \gets \Phi^{-1}(1-\frac{\delta}{2}) * \sigma / \mu$ \Comment{predicted error} \label{algo:ns-online:line:pred_error}


\EndWhile

\end{algorithmic}
\end{algorithm}

\cref{algo:ns-online} shows the \NS-online algorithm.
Each time we draw a sample (\cref{algo:ns-online:line:each-sampler}),
the only information we need to keep track of is the accumulated sum $\sum_{i=1}^n X_i$ (\cref{algo:ns-online:line:sum}) and the accumulated squared sum $\sum_{i=1}^n X_i^2$ (\cref{algo:ns-online:line:squaredSum}).
At the end of each interval (i.e. $W$ samples in \cref{algo:ns-online:line:period}), we compute the standard deviation $\sigma$ (\cref{algo:ns-online:line:std_dev}) and predict the error $\hat{\epsilon}$ (\cref{algo:ns-online:line:pred_error}).
If $\hat{\epsilon}$ is below the user's error bound (\cref{algo:ns-online:line:comp_error}), sampling is then terminated and it reports the estimated count $\mu$.

\subsection{\Eager for Neighbor Sampling}
\label{subsect:ns-prune}

A major drawback of \NS in prior systems is the low \hrate in dealing with sparse cases.
By looking at individual samples, we find that most of these samples fail to pass the closure check (\cref{algo:ns:line:closure} in \cref{algo:ns}). 
Therefore we looked deeper into the failures (i.e. missed samples). 
Our key observation is that many of the failures have been unpromising candidates even at the early stage of the sampling. 
For example, if we search for 6-\texttt{clique}s, and the first four vertices in the sample does not form a 4-\texttt{clique}, this sample is impossible to form a 6-\texttt{clique}.
Therefore, in general, the low hit rate in ASAP and Arya is because the closure check is delayed to the very last step, which we refer to as \textit{\lazy}.
Based on this understanding, we propose a \textit{\eager} approach to improve \NS performance. 
The key idea is to sample from promising candidates by verifying the pattern's connectivity constraints as early as possible. 
This strategy has two advantages. First, since unpromising candidates are pruned at early stage, and we only sample from promising candidates, each sample is more likely to succeed, i.e., hit a match. Second, if a sample starts from an edge whose neighborhood contains no or very few matches (unlikely to hit), 
\eager can minimize the work as this sample would fail at its early stage, while \lazy will have to proceed until the end and fail.  


The challenge in implementing \eager is how to avoid unpromising candidates but still retain unbiasedness in sampling.
Based on the subgraph tree abstraction (\cref{subsect:agpm}), our key finding is that each leaf in the tree corresponds to a unique path. As long as the pruning does not change this one-to-one mapping, we can prove that the sampler is unbiased.
We find that two typical pruning techniques, symmetry breaking~\cite{GraphZero} and matching order~\cite{G2Miner}, meet this requirement. 
There exist other pruning schemes in the literature~\cite{DecoMine,WCOJ}, 
which could be applied as well, but we leave this as a future work.
We refer \NS used in ASAP as \textit{\NS-base}, and \NS with the two pruning techniques as \textit{\NS-prune}. 
We first give an example to show how pruning avoids unpromising candidates, 
and then prove \NS-prune an unbiased estimator in \cref{lemma:ns-prune}.
\cref{algo:ns-prune-4cycle} shows the pseudo-code for finding \texttt{4-cycle} using \NS-prune.
In \cref{algo:ns-prune:line:v3} we compute a set intersection $N(v_0)$ \& $N(v_2)$ which is the candidate set of the third vertex $v_3$, because $v_3$ is a common neighbor of $v_0$ and $v_2$ in the \texttt{4-cycle} pattern.
In contrast, in \NS-base, because both $v_0$'s and $v_2$'s neighbors are possible candidates, $v_3$ is sampled from set union $N(v_0) \cup N(v_2)$, and in the final step it checks closure between $v_3$ and $v_2$ if $v_3$ is from $v_0$'s neighborhood, otherwise it checks closure between $v_3$ with $v_0$.
If the closure check fails, it is a miss.
However, in \NS-prune the closure check is unnecessary because $v_3$ is guaranteed to be connected to $v_0$ and $v_2$, 
and thus much more likely to hit a match.

\begin{algorithm}[t]
\footnotesize
\caption{\NS-prune for \texttt{4-cycle}}
\label{algo:ns-prune-4cycle}
\begin{algorithmic}[1]

\For{{\bf each} sampler $i \in$ [1, $N_s$]} \label{algo:ns:line:each-sampler}

\State $e(v_0, v_1) \gets$  \Call{sample}{$\mathcal{E}$} , $\alpha \gets 0$ \Comment{\hlg{sample} an edge $(v_0, v_1)$} \label{algo:ns-prune:line:sample-first-edge}
\State $A \gets$ $N(v_1)$ - \{$v_0$, $v_1$\}, bound by $v_0$ \Comment{\hlg{set difference}}  \label{algo:ns-prune:line:v2} 
\If{|A| = 0} break\EndIf
\State $v_2 \gets$ \Call{sample}{$A$} \Comment{\hlg{sample} node $v_2$ from set $A$} \label{algo:ns-prune:line:sample-v2} 

\State $B \gets$ $N(v_0)$ \& $N(v_2)$, bound by $v_1$ \Comment{\hlg{set intersection}} \label{algo:ns-prune:line:v3} 

\If{|B| = 0} break\EndIf
\State $v_3\gets$ \Call{sample}{$B$}  \Comment{\hlg{sample} node $v_3$ from set $B$} \label{algo:ns-prune:line:sample-v3}
\State $X_i \gets m * |A| * |B|$  \Comment{sampled count}
\EndFor
\end{algorithmic}
\end{algorithm}

\begin{lemma} \label{lemma:ns-prune}
  \NS-prune is an unbiased estimator.  
\end{lemma}

\begin{proof}
Let $(v_0,\dots,v_{k-1})$ be the vertices of an occurrence, i.e., a \textit{match}, of the pattern \pg in the matching order (e.g. for \texttt{4-cycle}, $(v_0,v_1),(v_1,v_2),(v_2,v_3),(v_3,v_0) \in E$). 
Each match corresponds to a unique $k$-tuple $(v_0,\dots,v_{k-1})$, e.g., for \texttt{4-cycle}, we enforce $v_0 = \max(v_0,v_1,v_2,v_3)$ and $v_3 < v_1$. During the execution of \NS-prune, the probability that $v_0$ and $v_1$ are chosen is $1/m$
(see \cref{algo:ns-prune:line:sample-first-edge} in \cref{algo:ns-prune-4cycle}), 
as all edges are equally likely. 
The probability that $v_2$ is chosen (see \cref{algo:ns-prune:line:sample-v2}) is $1/|A|$ as $v_2 \in A$. In general, the probability that $v_i$ is chosen is $1/|S_i|$, where $S_i$ is the candidate set that $v_i$ is drawn from. Therefore the probability that this particular match is sampled by \NS-prune is the product of these probabilities (e.g. $1/(m \cdot |A| \cdot |B|)$ for \texttt{4-cycle}). The scaling factor $\alpha$ is the inverse of this probability, so the expected contribution of this match to the estimated count for one sampler is $\frac{1}{\alpha} \cdot \alpha = 1$. 
Let $C$ be the true count of \pg in \dg, and let $X_{ij} = \alpha = (m \cdot |S_2| \cdots |S_{k-1}|)$ if the $i$th sample hits the $j$th match of \pg, 
otherwise $X_{ij}=0$, where $1 \le i \le N_s, 1 \le j \le C$. The estimated count is
$\mathbb{E}\left(\frac{1}{N_s}\sum_{i,j} X_{ij}\right) = \frac{1}{N_s}\left(\sum_{i,j} \mathbb{E}(X_{ij})\right) = \frac{1}{N_s}\sum_{i,j} 1 = C$.
Therefore \NS-prune is an unbiased estimator.
\end{proof}

\begin{figure}[t]
\centering
\includegraphics[width=0.36\textwidth]{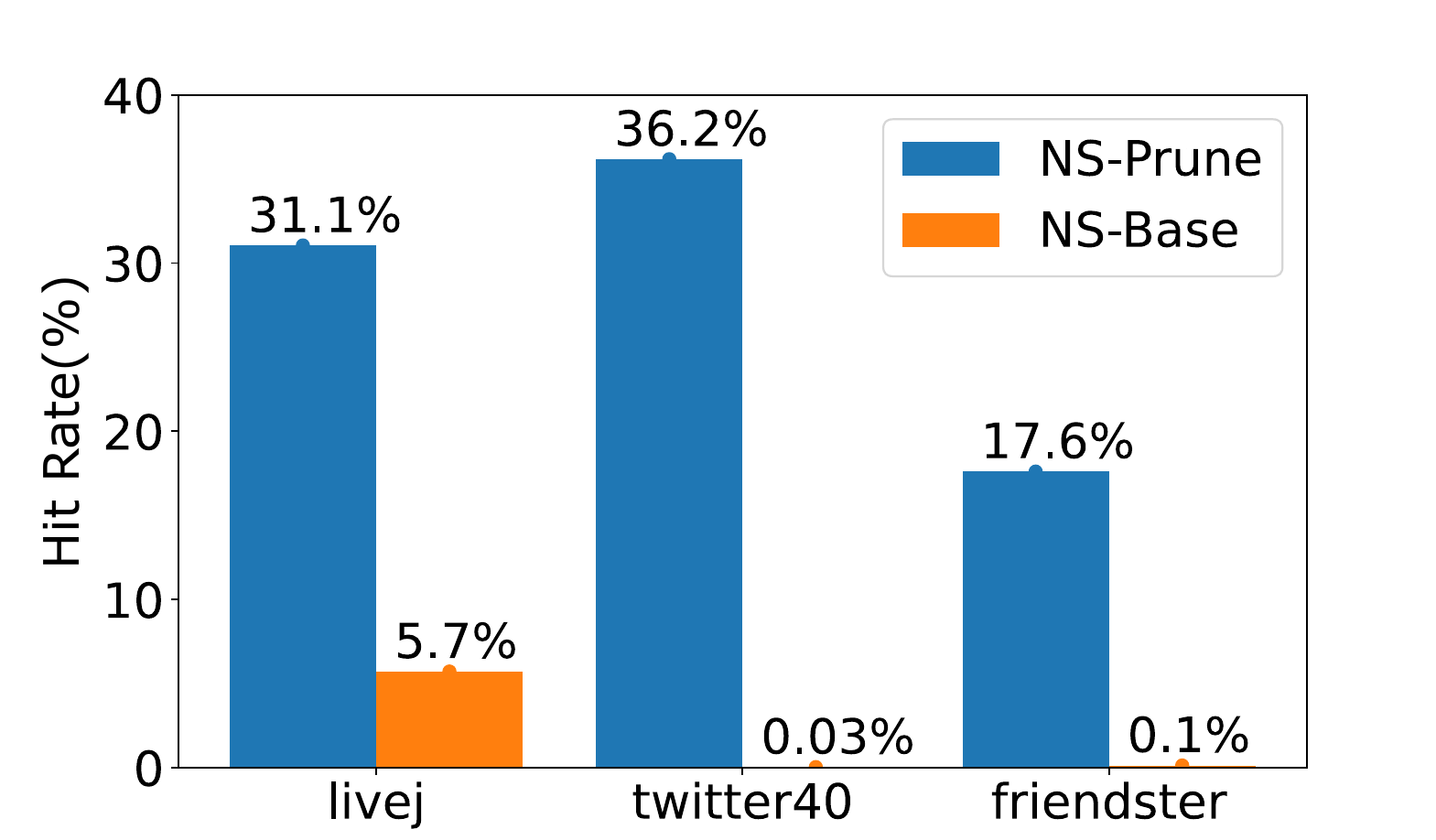}
\caption{Comparing the \hrate of \NS-prune with \NS-base, with the 4-\texttt{clique} pattern on various graphs. }
\label{fig:successrate}
\end{figure}

Note that the scaling factor in \NS-prune tends to be much smaller than that in \NS-base, because it involves the sizes of intersection instead of union sets. This results in lower variances and more stable and faster convergence, as we will show in \cref{sect:eval}.


\subsection{Cost Model for Neighbor Sampling}
\label{subsect:ns-perf-model}

In \NS, the total work is $\sum_{i=1}^{N_s}W_i$, where $W_i$ is the work of the $i$-th sample $s_i$. 
Given that $N_s \gg 1$ in the \NS scheme,
it is reasonable to assume that the \NS execution time is linear in the number of samples $N_s$.
So we can predict the execution time as $t_{tot}=t_{avg}*N_s=c*W_{avg}*N_s$, where $W_{avg}$ is the average work per sample, and $c$ is a hardware specific constant factor to translate work to time.
We estimate $N_s$ by profiling~\cref{subsect:profiler}.
As this is only used for performance prediction, it does not affect error and confidence.
To estimate $c$, a simple profiler can be run on each hardware machine to determine this constant scale factor. 
This profiling overhead is negligible, as it can be determined by running only once per machine or once per graph on only a small number of points.

The challenge, however, is to estimate $W_{avg}$ for the given \dg and \pg, 
as $W_i$ varies for different samples. 
In fact $W_i$ depends on the local neighborhood structure of $s_i$.
More specifically, the task of each sample is a sequence of set operations and sample operations, for example, see \cref{algo:ns-prune-4cycle}.
Given a specific pattern \pg, this sequence of operations is fixed, 
but each set operation in the sequence may take different (worse-case) time depending on the cardinality of the input sets, which overall depends on the structure of \dg.
So it is difficult to get an accurate estimation of $W_{avg}$, i.e., the slope the linear relationship, without really running \NS.





At each \texttt{break} point (e.g., \cref{algo:ns-base:line:breakA} in \cref{algo:ns-4cycle}), we must consider the possibility of an early-exit from the sampling procedure. Each \texttt{break} point is triggered based on the probability of being empty based on the candidate set size. However, it can be hard to predict this probability in practice, as certain graphs have many dense clustered areas, when sampled in few breakpoints occur, whereas others frequently terminate early.
To address this, we model the performance as a \textit{performance cone} (see example in \cref{fig:perf-predict-ns}) consisting of two slopes instead of one, where the performance is upper bounded by none of the breakpoints triggering (completing the full work of the sampling procedure) and lower bounded by the first \texttt{break} point. 

\subsection{Cost Model for Graph Sparsification} \label{subsect:gs-perf-model}

\begin{figure}[tb]
\centering
\includegraphics[width=0.41\textwidth]{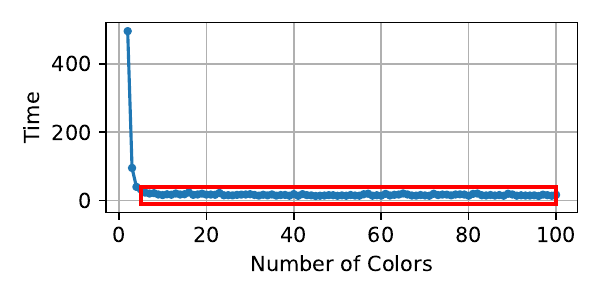}
\vspace{-6pt}
\caption{The execution time for 6\texttt{clique}-\texttt{Friendster} using Color Sparsification under different numbers of colors. Each point is one run. Red region is the stabilization window. }
\label{fig:color-speed}
\end{figure}



\GS running time contains two parts, one is the \textbf{preprocessing} time to generate the sparsified graph $G'$, the other is the time spent on \textbf{exact search} in the sparsified graph $G'$.

Preprocessing involves looping over every edge in order to remove edges.
Therefore, the total work is $|E| \cdot (w_1+p\times w_2)$, where $p$ is the probability that an edge is kept, $w_1$ is the work on every edge, $w_2$ is the work on kept edges.
In our implementation $w_1$ is one read operation, $w_2$ is one write operation.

For exact search in $G'$, our estimation is again based on set operations.
But instead of a sequence of operations in \NS,
the work in \GS consists of nested loops, 
each of which corresponds to one vertex in the pattern and iterates over the candidates of that vertex.
The candidate vertex set is computed by set operations.
If we describe a GPM algorithm as a sequence of nested for loops $M = (X_1, \dots ,X_n)$, which refine the candidate set to a possible match. Each nested loop can be described as $X_j = (o_j, i_j)$. Where $o_i$ is the number of operations performed in the inner body of that loop to refine the candidate set. $i_j$ is the number of iterations for the $j$th loop. $i_j = |Z_i|$, the size of the candidate set at the $j$th level. $o_i$ is the work of the operations to generate the candidate set $|Z_{i+1}|$.
Then the total work can be described as
$ \prod_{j}^{i_{j}}o_j = \prod_{j}^{|Z_{j}|}o_j $.
To estimate $i_j$ and $o_j$, we need to estimate the cardinality of each candidate set $Z$.

If $Z=A - B$, i.e., set difference, $|Z|$ is bounded by $|A|$, which can be simply estimated as the average degree, $|V|/|E|$. 
If $Z=A \cap B$, i.e., set intersection, 
$|Z|$ is bounded by $\min(|A|, |B|)$.
To get a better bound, we can use the method in GraphPi~\cite{GraphPi} and estimate $|Z|$ as $|V| \cdot p_1 \cdot p_2^{n-1}$, 
where $p_1 = \frac{2\cdot | E|} {|V|^2}$, 
$p_2 = \frac{T \cdot |V|}{2 \cdot |E|^2}$ and $T$ is the triangle count in \dg.
Intuitively, $p_1$ is the probability of any pair
of vertices being neighbors.
$p_2$ is the probability of any pair of vertices in a
neighborhood being directly connected to each other. 
To use this estimation for \GS, we need to make the following adjustments.
Note that in a sparsified graph $G'$, $|V'| = |V|$ and $|E'| = |E|\times p$, 
and we have $T'$ = $T\times p^2$ as discussed in \cref{subsect:subgraph-sampling}.
Then if $Z=A - B$, $|Z|$ is estimated as $\frac{|V|}{|E| \cdot p}$.
If $Z=A \cap B$, $|Z|$ is estimated as $|V| \cdot p_1 \cdot p_2^{n-1}$, where $p_1 = \frac{2\cdot | E|\cdot p} {|V|^2}, p_2 = \frac{T \cdot |V|}{2 \cdot |E|^2 \cdot p^4}$.


\section{System Design and Implementation}

We give an overview of the \sys system in \cref{subsect:sys-overview}, describe details of the \GS engine in \cref{subsect:treadoff-gs} and our proposed profiling mechanism in \cref{subsect:profiler}, 
and other implementation details in \cref{subsect:impl}.

\begin{figure}[tb]
\centering
\includegraphics[width=0.49\textwidth]{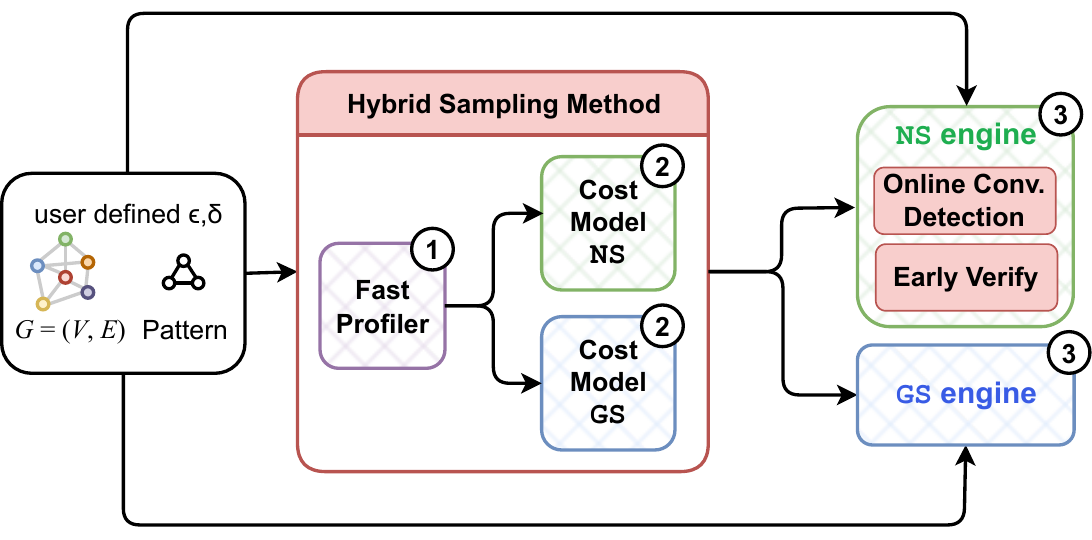}
\caption{\sys system overview. 
The three red boxes are our proposed novel techniques: online convergence detection, early pruning and hybrid sampling.
\NS: neighbor sampling. \GS: graph sparsification. 
The system execution flow is
{\ding{172} fast profiler estimates input parameters (e.g., \texttt{\#colors}, \texttt{\#samples}), 
\ding{173} cost models predict performance and select from \NS and \GS sampling schemes,
and \ding{174} the selected (\NS or \GS) engine is invoked to conduct sampling. }}
\label{fig:sys-overview}
\end{figure}

\vspace*{-3pt}
\subsection{System Overview and Interface} \label{subsect:sys-overview}

\cref{fig:sys-overview} illustrates the major components in our system. \sys is composed of a fast profiler, two cost models and two execution engines for \NS and and \GS respectively. 
The cost models have been described in \cref{subsect:ns-perf-model} and \cref{subsect:gs-perf-model} respectively.
The \NS engine is enhanced with our novel convergence detection mechanism (\cref{subsect:converge}) and is significantly improved by our proven unbiased optimizations (\cref{subsect:ns-prune}).
Our \GS engine is the first generalized color sparsification for arbitrary patterns,
as all prior \GS-based work are customized for specific patterns, e.g., triangle.
To generalize the \GS approach for an arbitrary pattern \pg, we need to (1) generate a pattern specific exact search program, (2) determine the scaling factor for \pg and (3) determine the values of its key parameter, i.e., the sparsify probability $p$ or the number of colors $c$  ($c=\frac{1}{p}$).
For (1) we can leverage the state-of-art compiler based approach~\cite{GraphZero}. 
For (2) we explain in \cref{subsect:treadoff-gs}.
For (3) our fast profiler in \cref{subsect:profiler} determines it.

To meet various accuracy requirements in applications, \sys provides two modes for the user to choose, \textit{strict mode} (default) and \textit{loose mode}.
The {strict mode} uses only the \NS engine
with online convergence detection to guarantee high confidence. 
In this mode the fast profiling and cost models are bypassed.
The {loose mode}, however, employs our proposed hybrid approach.
It uses the fast profiler and cost models to determine if the \GS or \NS engine is used. 
When comparing the predicted performance of \NS and \GS, \sys uses a thresholding mechanism to check if the predicted performance of \GS overlaps with the performance cone of \NS, i.e., it is either entirely above or entirely below the cone.
If not, \sys chooses the faster one and activates the corresponding engine.
Otherwise, the two schemes should perform similarly well, hence \sys chooses \NS to guarantee high confidence.




\subsection{Tradeoff in the \GS Engine} \label{subsect:treadoff-gs}

For simplicity, we discuss edge sparsification, while color sparsification is similar. 
The estimated count is $\hat{C} = Y\cdot p^{-l}$, where the random variable $Y$ denotes the number of matches in the sampled graph $G'$, which is sparsified from \dg with probability $p$. For each match $M$ in \dg, the probability that $M$ exists in $G'$ is $p^l$, hence the expected number of matches in $G'$ is   $\mathop{\mathbb{E}}[Y] = p^lC$, i.e., $\mathop{\mathbb{E}}[Y\cdot p^{-l}] = C$, so we have $\mathop{\mathbb{E}}[\hat{C}] = C$.  
Although the edges are sampled independently, the matches are not.
Consider an edge $e$ shared between two matches $M_i$ and $M_j$ of \pg. When $e$ gets removed during sparsification, necessarily both $M_i$ and $M_j$ will not be counted.
So the Chernoff bounds do not apply directly for \GS.
We can use Chebyshev’s bounds instead.
$\mathop{\mathbb{V}ar}[\hat{C}]$ $= C \cdot (p^{-l}-1)+\mathop{\mathbb{C}ov}$, where $\mathop{\mathbb{C}ov}=$
$\sum_{z=2}^{k-1}t_z\cdot (p^{1-z}-1)$ and $t_z$ is the number of pairs of matches that share $z$ vertices.
The variance depends on both the number of matches $C$ in \dg and the number of pairs of matches that share one or more edges. 

Apparently, the key knob to tune accuracy is $p$.
As we increase $p$, the variance decreases and accuracy increases.
Since \GS speed is insensitive to the number of colors $c=\frac{1}{p}$ within a wide stabilization window (\cref{fig:color-speed}), we can pick a large $p$ to achieve better accuracy.
Note that the variance increases exponentially in the pattern size $k$.
This means larger patterns are more difficult to estimate accurately,
and so we may need a larger $p$, to guarantee the same error bound.

\subsection{Fast Profiling for Cost Models} \label{subsect:profiler}

To predict performance in our cost models, we need the number of samples $N_s$ for \NS and the number of colors $c$ for \GS.
Our fast profiler is used to quickly determine the values of these parameters.
Note that $N_s$ here is only used in the \NS cost model, not used for \NS sampling (as our \NS-online does not need to predict $N_s$). 
The profiler first generates a sparsified graph $G'$, e.g., 10\% of the original graph.
Then, it runs our \NS-online engine with an internal error bound (50\%) and confidence {(99\%)}. 
By detecting convergence, we can determine $ N_s = \frac{{N_o} \cdot \hat{ \epsilon} \cdot {\mu} \cdot \rho(P,G)}{\overline{S} \cdot \epsilon ^2 \cdot \rho(P,G')} $, where
$N_o$ is the number of samples \NS-online converged with,  
$\mu$ is the count returned by \NS-online, 
$\overline{S}$ is the scaled count from $G'$ to $G$ (see \ref{subsect:subgraph-sampling}), 
{$\hat{\epsilon}$ is the final predicted error by \NS-online},
$\rho(P,G)$ is the probability of sampling pattern $P$ in $G$, which is used in prior systems \cite{ASAP,Arya}, and determined by properties of $G$, e.g., $\Delta$ and $|E|$. 

\subsection{Parallel Implementation Details} \label{subsect:impl}

Our \NS and \GS engines are both parallelized. 
In \GS, because sparsification creates non-overlapping subgraph partitions, each partition can be searched independently. 
Within each partition, we further parallelize it over each vertex in the subgraph, which provides enough parallelism.
In \NS, as each sample is an independent task, it can be embarrassingly parallelized across samples.
Note that our \NS-online requires a barrier synchronization at the end of each interval $W$ (\cref{algo:ns-online:line:period} in \cref{algo:ns-online}). 
If $W$ is too small, we have limited parallelism and too much synchronization overhead. 
But if $W$ is too large, we may end up with drawing more samples than necessary.
Thus, we want $W$ to be large enough to maximize parallelism and minimize synchronization overhead, and small enough to minimize redundant work.
We set $W$ to be a fixed percent of the estimate of $N_s$ returned from the fast profiler (e.g. 10\% of $N_s$). 


\section{Evaluation} \label{sect:eval}

\begin{table}[t]
\centering
\resizebox{0.48\textwidth}{!}{
	\begin{tabular}{crrrrrr}
		\Xhline{2\arrayrulewidth}
		\bf{Graph} & \bf{Source} & {\bf{|V|}} & {\bf{|E|}} & {\bf{Avg deg.}} & {\bf{Max deg.}}\\
		\hline
		\texttt{lj} & {liveJournal}~\cite{SNAP} & 4.8M & 43M & 17.7 & 20,333\\
		\texttt{tw} & {twitter40}~\cite{twitter40} & 42M & 2.4B & 57.7 & 2,997,487\\
		\texttt{fr} & {friendster}~\cite{friendster} & 66M & 3.6B & 55.1 & 5,214\\
		\texttt{uk} & {uk2007}~\cite{uk2007} & 106M & 6.6B & 62.4 & 975,419\\
		\texttt{gsh} & {gsh-2015}~\cite{gsh15} & 988M & 51B & 52.0 & 58,860,305 \\
		\texttt{cw} & {clueweb12}~\cite{gsh15} & 978M & 75B & 76.4 & 75,611,696 \\
		\Xhline{2\arrayrulewidth}
	\end{tabular}
}
\caption{\small Data graphs (symmetric, no loops or duplicate edges). Maximum degrees are smaller when orientation is applied for cliques.}
\label{tab:input}
\end{table}


We implement \sys in C++ and OpenMP for parallelization. 
We use \sys-\NS, \sys-\GS, \sys-{\tt HY} to represent the \NS, \GS, and hybrid mode of our system respectively.
In this evaluation we focus on comparing with prior A-GPM systems that provide both generalization and automation (see \cref{sect:related} for discussion on non-systematic solutions). 
We compare \sys with the state-of-the-art A-GPM system, Arya~\cite{Arya}, and exact GPM system, GraphZero~\cite{GraphZero}. 
We do not include ASAP since Arya always outperforms ASAP.
We test on a 3.0 GHz, 48-core (2-sockets, 24 cores per socket) Intel CPU without hyperthreading, with up to 1TB of memory. 
\cref{tab:input} shows the graphs used in our experiments, which are representative real-world graphs with varying sizes and topology characteristics.
In \sys, graphs are represented in the Compressed Sparse Row (CSR) format.
We evaluate two types of GPM tasks: subgraph counting (edge-induced) and motif counting (vertex-induced). 
For subgraph counting, we test on patterns including $k$-\texttt{\small cliques}, 5-\texttt{\small path}, \texttt{\small house}, and \texttt{\small dumbbell}.
We do not include even larger non-clique patterns, because we can not verify the errors as their exact counts are unknown for most of the graphs.
In all the experiments, we time out at 10 hours.
{We then conservatively use 10 hours for the timed-out cases when calculating speedups.}

We first compare the overall performance of \sys with state-of-the-art systems in \cref{subsect:overall-perf}. We show how our convergence detection performs in \cref{subsect:converge-perf}, and the accuracy the \NS and \GS cost models in \cref{subsect:model-accuracy}. We discuss system efficiency
in \cref{subsect:sys-efficiency}.

\subsection{Sampling Performance vs. State-of-the-Art} \label{subsect:overall-perf}

We compare sampling performance with Arya and GraphZero.
For the hybrid mode we discuss its profiling time in \cref{subsect:sys-efficiency}. 
For \sys-\GS, we do include the preprocessing time spent on sparsifying the data graph.
We use an error bound of 10\% and confidence of 99\%, which is the common practice in ASAP and Arya. 

\begin{table*}[tp]
\resizebox{0.75\textwidth}{!}{
\begin{tabular}{c|r|rrr|rrr|rrr}
\multicolumn{2}{c|}{\textbf{Pattern}} & 
\multicolumn{3}{c|}{3-\textbf{clique} (triangle)} & \multicolumn{3}{c|}{\textbf{4-clique}} & \multicolumn{3}{c}{\textbf{6-clique}}    \\ 
\hline
\multicolumn{2}{c|}{\textbf{Graph}} & \texttt{Lj} & \texttt{Tw} & \texttt{Fr} & \texttt{Lj} & \texttt{Tw} & \texttt{Fr} & \texttt{Lj} & \texttt{Tw} & \texttt{Fr} \\
\hline

& time (sec) & \textbf{0.001}  & \textbf{0.046} & 0.026  & \textbf{0.003}  & \textbf{0.068}  & \textbf{0.090}  & \textbf{0.059}  & \textbf{0.707} & \textbf{1.132} \\
{\sys}-\NS & hit rate & 18\% & 55\% & 53\% & 7.3\% & 46\% & 31\% & 6.4\% & 46\% & 16\% \\
 & \# samples & 1e4 & 2e4 & 1e4 & 8e4 & 2e5 & 5e4 & 2.1e6 & 4.1e6 & 7.7e7 \\
\hline

& time (sec) & 0.014 & 1.193 & \textbf{0.017} & 94.2 & 9656.9 & 2057.7 & TO & TO & TO \\
Arya & hit rate & 10.6\%  & 3.9\%  & 2.9\% & 3.2e-3\% & 2.7e-3\% & 1.7e-6\% & --  & -- & -- \\
& \# samples & 4.3e4  & 6.3e4 & 4.6e4  & 3.3e8  & 2.6e8  & 5.1e9 & $\times$  & $\times$ & $\times$ \\
\hline

GraphZero & time (sec) & 0.434 & 58.0 & 83.0 & 2.1 & 4004.5 & 73.1 & 2502.5 & TO & 160.4 \\
\Xhline{2\arrayrulewidth}
\end{tabular}
}
\caption{\small $k$-clique performance (10\% error). 
TO: timed out. 
Arya uses ELP. 
$\times$: ELP does not converge.
Fastest time in bold.
}
\label{tab:clique}
\end{table*}

\cref{tab:clique} compares the $k$-\texttt{clique} ($k$=3,4,6) running time of \sys with Arya and GraphZero.
Overall cliques, \sys achieves a \texttt{geomean} average speedup of \usoveraryaclique (up to \usoveraryacliquemax) against Arya, and \usoverexactclique over (up to \usoverexactcliquemax) GraphZero respectively.
In general, the speedup of \sys-\NS over Arya comes from two parts.
First, our online convergence detection gives stable and precise termination condition (demonstrated later in \cref{fig:convergence} and \cref{fig:stable-estimated-error}), while ELP in Arya could suggest very conservative termination conditions that do way more samples than necessary (see 
\cref{fig:sample-reduction}).
Second, our \NS-prune approach dramatically improves the \hrate over \NS-base and thus the total number of samples is reduced.
This is evidenced in \cref{fig:successrate} and further confirmed in \cref{fig:ns-base-ns-prune}.
The significant speedups over Arya are expected because (1) Arya is based on pattern decomposition and cliques are hard to decompose, (2) cliques are more likely to fall in the needle-in-the-hay cases which Arya handles poorly.
The speedups are further evidenced by the hit rate and number of samples $N_s$ used.  
{In particular, for 4-\texttt{clique}, Arya requires 3 to 5 orders of magnitude more samples,
while its hit rates are extremely low, e.g., {1.7$\times 10^{-6}$\%} for \texttt{Fr}.
}
Notably, Arya is {28$\times$} slower than GraphZero for the sparse case 4-\texttt{clique} on \texttt{Fr}. 
Moreover, Arya's ELP can not converge within 10 hours for the even more sparse case, 6-\texttt{clique} on \texttt{Fr},
emphasizing the limitation of ELP.
{For \texttt{triangle} on \texttt{Fr}, 
\sys-\NS is slightly slower than Arya due to synchronization overhead.
}

\begin{table*}[t]
\resizebox{.78\textwidth}{!}{
\begin{tabular}{c|r|rrr|rrr|rrr}
\multicolumn{2}{c|}{\textbf{Pattern}} & \multicolumn{3}{c|}{\textbf{5-path}} & \multicolumn{3}{c|}{\textbf{5-house}} & \multicolumn{3}{c}{\textbf{6-dumbbell}}    \\ 
\hline
\multicolumn{2}{c|}{\textbf{Graph}} & \texttt{Lj} & \texttt{Tw} & \texttt{Fr} & \texttt{Lj} & \texttt{Tw} & \texttt{Fr} & \texttt{Lj}& \texttt{Tw}  & \texttt{Fr}\\
\hline

& time (sec)  & \textbf{0.004}  & \textbf{0.91} & \textbf{0.10}  & \textbf{0.008} & \textbf{944.6*} & \textbf{0.159} & \textbf{0.017}  & \textbf{159.3}&\textbf{0.331}  \\

\sys-\NS & hit rate & 29\%  & 99\%  & 90\%   & 43\% & 43\% & 15\% & 16\% & 43\% & 32\% \\
& \# samples & 3e4 & 1e4 & 2e4 & 5e4 & 1.9e7 & 3.3e5 & 1.5e5 & 1.6e7 & 6.9e5 \\
 \hline

& time (sec)  & 4.422 & 13.78 & 77.93 & 395.6 &  1297.4& TO & 346.2   &4644.4 & TO \\
Arya & hit rate & 0.04\%  & 5.8e-03\%  & 1.7\%   & 2.6e-3\%   & 0.02\%   & --   & 0.04\%  & 2.2e-3\% & -- \\
 & \# samples & 1.1e7 & 1.4e7 & 3.4e5  & 1.5e9 & 6.8e7 & $\times$ & 2.1e8 & 1.1e9 & $\times$ \\
\hline

GraphZero & time (sec) & 262.1 & {TO} & {TO} & TO & {TO} & {TO} & {TO} & {TO} & {TO}  \\
\Xhline{2\arrayrulewidth}
\end{tabular}
}
\caption{\small Non-clique pattern 
performance (10\% error).
TO: timed out. 
$\times$: ELP does not converge. 
* \GS is selected by our hybrid method.
}
\label{tab:complex}
\end{table*}

Table~\ref{tab:complex} compares performance on large, non-clique patterns, including 5-\texttt{path}, 5-\texttt{house} and 6-\texttt{dumbbell}. 
Note that it is well known that for exact GPM solvers (e.g. GraphZero) it is much more expensive to search for these patterns than cliques, as they are sparser and their sizes are beyond 4.
We observe that GraphZero is timed out for most of the cases, which shows that it is critical to use approximation for large (sparse) patterns.
Arya enjoys fairly good speedups over GraphZero.
This is because, compared to cliques, these patterns are easier to decompose, which is the case that favors Arya.
However, for \texttt{Fr}, Arya still suffers extremely high $N_s$ and low hit rate, and hence runs more than 10 hours for 5-\texttt{house} and 6-\texttt{dumbbell}.
In all these cases, \sys significantly improves hit rates and reduces $N_s$, which lead to fast convergence speed. 
Across these patterns, \sys achieves an {\texttt{geomean}} average of \usoveraryalargepattern and \usoverexactlargepattern speedup against Arya and GraphZero respectively.
Note that for \texttt{Tw} on \texttt{house}, as shown in \cref{tab:gs}, \sys-\texttt{HY} selects the \GS engine in this case, since \GS is predicted to be faster, based on our prediction on the number of samples and time per sample.

\begin{table*}
\captionsetup[subtable]{aboveskip=3pt,belowskip=3pt}
\centering
 \begin{subtable}{.475\linewidth}
 \centering
 \resizebox{\textwidth}{!}{
\begin{tabular}{r|rrr|rrr}
\multicolumn{1}{c|}{\textbf{Pattern}} & \multicolumn{3}{c|}{\textbf{3-motif}} & \multicolumn{3}{c}{\textbf{4-motif}}   \\ 
\hline
\multicolumn{1}{c|}{\textbf{Graph}} & \texttt{Lj} & \texttt{Tw} & \texttt{Fr} & \texttt{Lj} & \texttt{Tw} & \texttt{Fr}\\
\hline

\sys-\NS  & \textbf{0.004} & \textbf{0.5} & 0.14 & \textbf{ 0.06}  & \textbf{82.7}  & \textbf{0.2} \\

Arya &0.020 & 1.9 & \textbf{0.08} & 231.63 & 13180.2 & 6157.9 \\
GraphZero & 1.283 & 16316.2 & 242.62 & 1927.27 & {TO} &  TO \\
\Xhline{2\arrayrulewidth}
\end{tabular}
}
 \caption{\small Performance on Motif Counting. }
 \label{tab:motif}
 \end{subtable}
\quad
 \begin{subtable}{.495\linewidth}
 \centering
 \resizebox{\textwidth}{!}{
\begin{tabular}{r|rrr|rrr|rrr}
\multicolumn{1}{c|}{\textbf{Pattern}} & \multicolumn{3}{c|}{\textbf{triangle}} & \multicolumn{3}{c|}{\textbf{4-clique}} & \multicolumn{3}{c}{\textbf{5-path}*}    \\ \hline
\multicolumn{1}{c|}{\textbf{Graph}} & \texttt{uk} & \texttt{gsh} & \texttt{cw} & \texttt{uk} & \texttt{gsh} & \texttt{cw} & \texttt{uk} & \texttt{gsh} & \texttt{cw} \\ \hline
\sys-\NS & \textbf{0.09} & \textbf{0.28}  & \textbf{0.19} & \textbf{0.11}& \textbf{0.88}  & \textbf{1.09} & \textbf{0.1}& \textbf{11.2} &\textbf{156.0} \\
Arya & 0.2 & 0.7& OoM& 175.5 & $\times$ & OoM & 4.2 & $\times$ & OoM  \\
GraphZero & 73.3 & 153.6 & 198.2 & {TO} & {TO} & {TO} &{TO}& {TO} & {TO} \\
\Xhline{2\arrayrulewidth}
\end{tabular}
}
\caption{\small Performance on huge graphs. 
}
\label{tab:large}
 \end{subtable}
\vspace{-5pt}
\caption{Running time (sec) (10\% error). 
TO: timed out. 
$\times$: ELP does not converge. 
* the true count is unknown, so accuracy is not verified. 
}
\label{tab:combine-motif-large}
\end{table*}

Table~\ref{tab:motif} reports the motif counting performance. 
Note that in motif counting, we look for vertex-induced subgraphs, unlike cases in \cref{tab:complex} which search for edge-induced subgraphs.
The key difference in computation is that we need both set intersection and set difference to find vertex-induced subgraphs, but only set intersection is needed for edge-induced subgraphs.
Despite more computation needed, \sys is still significantly faster than Arya.
Across motifs, \sys achieves a \texttt{geomean} speedup of \usoveraryamotif speedup against Arya,
and \usoverexactmotif speedup over GraphZero.

\cref{tab:large} compares performance on huge graphs, i.e., {uk2007 (\texttt{uk})} gsh-2015 (\texttt{gsh}) and clueweb12 (\texttt{cw}). 
Note that Arya mostly runs out of memory for \texttt{cw} because (1) it has to maintain the set union results (for each of the parallel threads) in memory, and (2) its internal representation of the graph stream is implemented in COO-like format,
which is less compact than CSR. 
\sys achieves a \texttt{geomean} average \usoveraryalargegraph compared to Arya and a \usoverexactlargegraph speedup compared to GraphZero. 

Table ~\ref{tab:gs} shows the cases where \sys-{\tt HY} selects the \GS mode. Notably, for \texttt{Fr} and $k=9$ (sparse graph and big dense pattern) which is a needle-in-the-hay case for \NS, \sys successfully choose to use \GS instead of \NS.
This switch from \NS to \GS brings us a \gsovernsmax performance improvement. Over all the cases where the \GS engine is selected, the \GS engine leads to a \texttt{geomean} average \gsoverns speedup over the \NS engine.

\begin{table*}
\captionsetup[subtable]{aboveskip=3pt,belowskip=3pt}
\centering
 
 \begin{subtable}{.395\linewidth}
 \centering
 \resizebox{\textwidth}{!}{
\begin{tabular}{r|rr|rr|r}
\multicolumn{1}{c|}{\textbf{Pattern}} & \multicolumn{2}{c|}{\textbf{8-\texttt{clique}}} & \multicolumn{2}{c|}{\textbf{9-\texttt{clique}}} & \multicolumn{1}{c}{\textbf{5-\texttt{house}}}    \\ \hline
\multicolumn{1}{c|}{\textbf{Graph}} & \texttt{lj} & \texttt{fr} & \texttt{lj} & \texttt{fr} & \texttt{tw}  \\ \hline
\sys - \NS  & 3.0  & 1663.6  & 17.4  & 8358.0 &944.6  \\
\sys - \GS & \textbf{0.7}  &  \textbf{43.4}  &   \textbf{2.1} & \textbf{134.9} & \textbf{200.1} \\
\Xhline{2\arrayrulewidth}
\end{tabular}
}
 \caption{\small running time (sec) of \NS and \GS} 
 \label{tab:gs}
 \end{subtable}
\quad
 \begin{subtable}{.53\linewidth}
 \centering
 \resizebox{\textwidth}{!}{
\begin{tabular}{r|rrrr|rrrr}
\multicolumn{1}{c|}{\textbf{Pattern-Graph}} & \multicolumn{4}{c|}{\textbf{\texttt{house} -- \texttt{Lj}}} & \multicolumn{4}{c}{\textbf{4-\texttt{clique} -- \texttt{Tw}}}   \\ 
\hline
\multicolumn{1}{c|}{\textbf{Error bound}} & \textbf{10\%} & \textbf{5\%}  & \textbf{2\%} & \textbf{1\%} & \textbf{10\%} & \textbf{5\%}  & \textbf{2\%} & \textbf{1\%}\\
\hline

\sys-\NS   & \textbf{0.008} & \textbf{0.053} & \textbf{0.3}  & \textbf{1.1}  & \textbf{0.07} & \textbf{0.13} & \textbf{0.69} & \textbf{2.68} \\

Arya &395.6 & 1639.6 & 9873.8 & TO & 9656.8 & TO & TO & TO \\

\Xhline{2\arrayrulewidth}
\end{tabular}
}
 \caption{\small running time (sec) under error bounds from 10\% to 1\%}
 \label{tab:error-varied}
 \end{subtable}
\vspace{-5pt}
\caption{\small (a) Comparing \NS and \GS performance when \sys-{\tt HY} selects \GS; (b) performance changes when varying error bounds.}
\end{table*}

Table ~\ref{tab:error-varied} shows Arya and \sys running time with varied error bounds. 
We observe that the performance gap between Arya and \sys remains huge (ranging from 31 to 49 thousand times) as we decrease the error bound.

Overall, we achieve \usoveraryaall speedup over Arya, and four orders of magnitude speedup over GraphZero.

\vspace{-3pt} 
\subsection{Effectiveness of Convergence Detection} \label{subsect:converge-perf}

\begin{figure}[tb]
\includegraphics[width=0.49\textwidth]{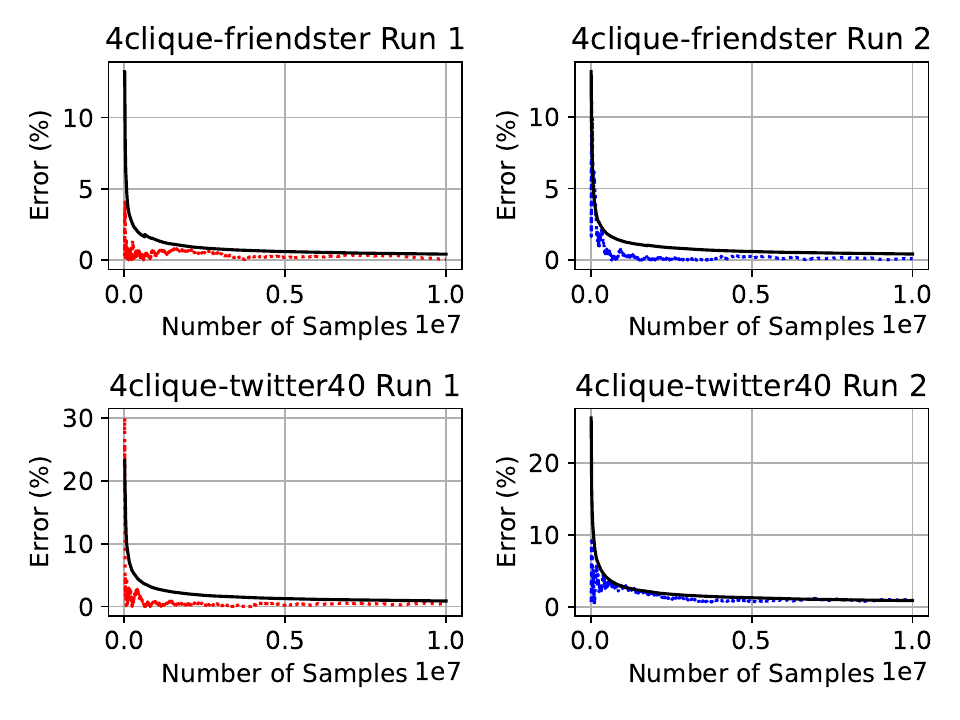}
\caption{Comparing \sys-\NS's estimated errors (black curves) with actual errors (red and blue curves). 
We 
do two runs of 4-\texttt{clique} counting on two graphs \texttt{Fr} and \texttt{Tw}. }
\label{fig:convergence}
\end{figure}

\cref{fig:convergence} compares the error predicted by \sys-\NS with the actual error throughout the sampling procedure, to show the effectiveness of \sys-\NS 's online convergence detection.
We illustrate two cases here, but we have verified the same trend for all our test cases in the evaluation.
In \cref{fig:convergence} we observe that for both cases, the predicted error curves strictly bound the actual error, which verifies the high confidence that our method achieves.
In addition to our theoretical proof in \cref{subsect:converge}, this empirical study further demonstrates that our method provides strong guarantee on confidence,
which is critical in applications where users want to have strict accuracy requirement.

\begin{figure}[tb]
\centering
\includegraphics[width=0.38\textwidth]{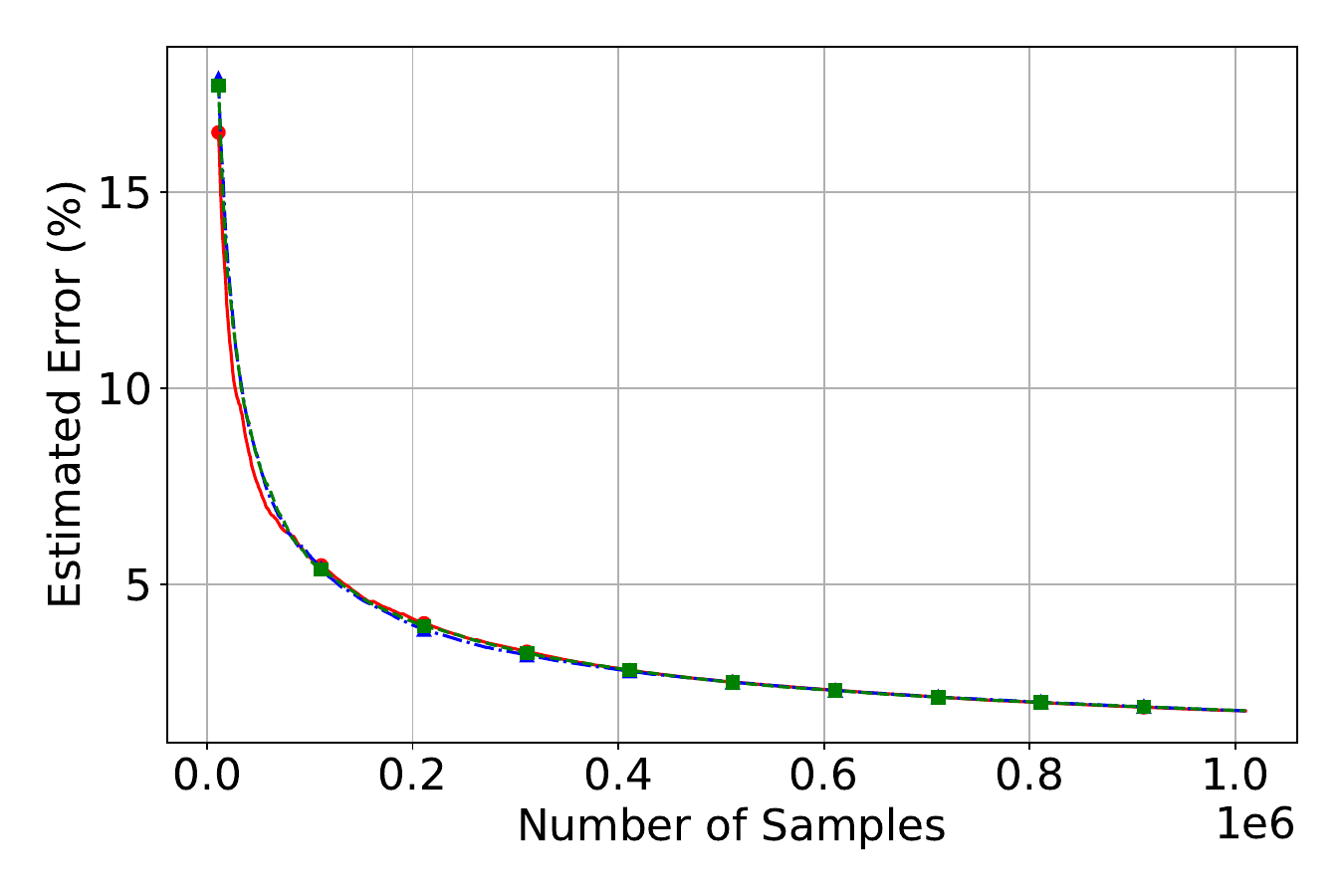}
\vspace{-5pt}
\caption{Predicted errors in our convergence detection mechanism across three different runs (\texttt{4-clique} on \texttt{Lj}) are extremely stable.}
\label{fig:stable-estimated-error}
\end{figure}

\cref{fig:stable-estimated-error} shows the stability of the error estimation method in \sys-\NS. 
We do three repeated runs on the \texttt{4-clique} pattern and \texttt{Lj} graph.
In contrast to the unstable predictions (\cref{fig:elp-unstable}) given by ELP in prior systems, our estimated errors are almost the same across three independent runs. 
{Moreover, our online mechanism never fails, while ELP suffers convergence issue that may lead to endless preprocessing (e.g., for 6-\texttt{clique} in \cref{tab:clique}). } We observe the same trend in stability on other graphs and patterns.
This experiment further demonstrates that our method is highly reliable and can be adopted in practice.

\begin{figure}[tb]
\centering
\includegraphics[width=0.4\textwidth]{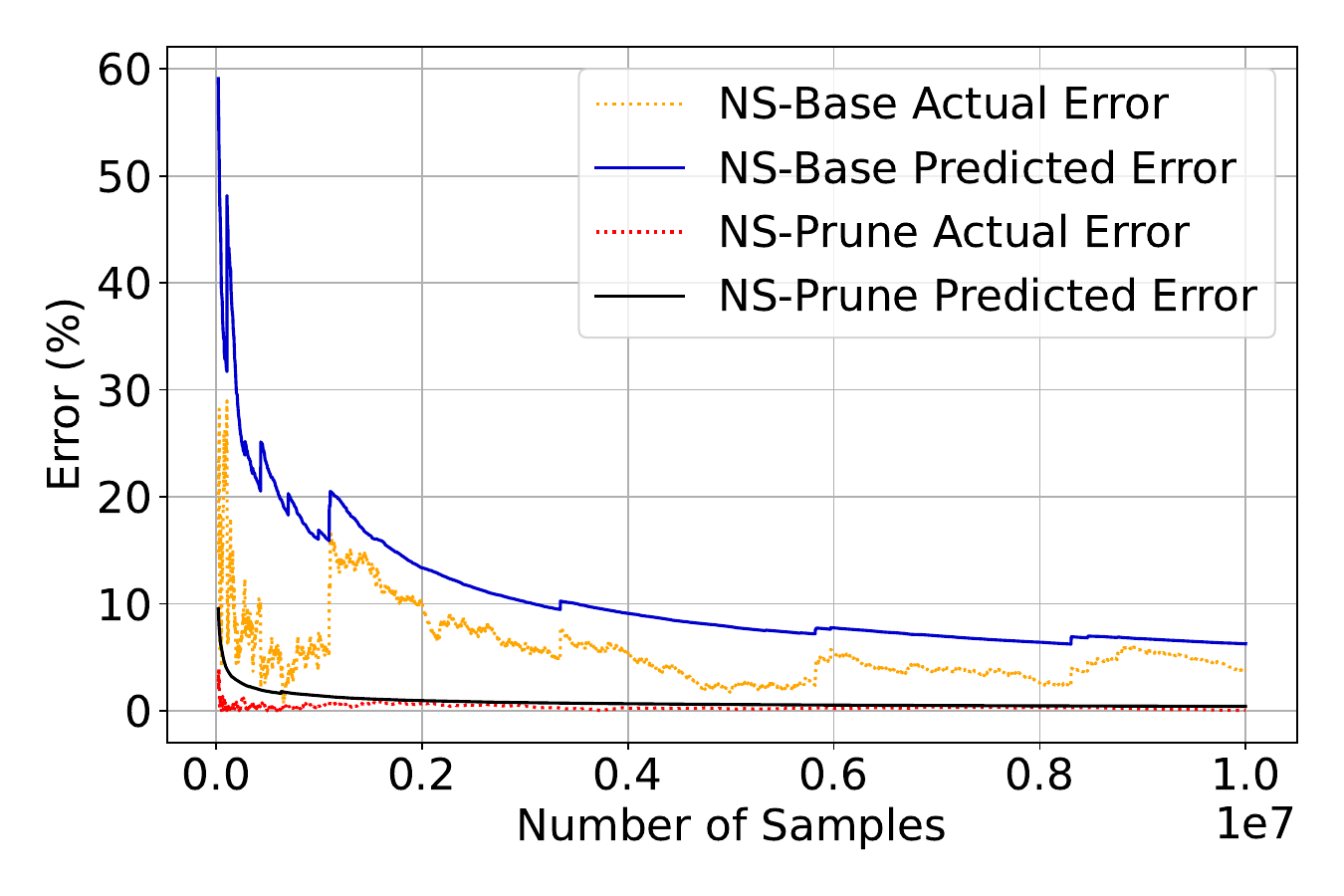}
\vspace{-5pt}
\caption{Comparison between the convergence rate of \NS-prune and \NS-base. The same formula (with two standard deviations of confidence) was used to generate the error estimate curves.}
\label{fig:ns-base-ns-prune}
\end{figure}

\cref{fig:ns-base-ns-prune} compares the convergence rate of \NS-prune and \NS-base, both under our online detection framework. 
We observe that our proposed early-pruning mechanism employed in \NS-prune result in roughly an order of magnitude lower error for the same number of samples.
Therefore, with pruning, our system can converge much faster that ASAP and Arya, 
which is the other major reason why \sys achieves much better performance.

\cref{fig:sample-reduction} shows the reduction on the number of samples $N_s$, by incrementally applying online convergence detection (orange) and \eager (green), against ELP in Arya (blue).
Note that in \cref{tab:clique} and \cref{tab:complex} we report $N_s$ only for \NS-online-prune, but here we breakdown the contributions of online detection and pruning.
We observe that our online method reduces $N_s$ sharply over Arya.
Applying pruning in \eager further reduces $N_s$ by a significant amount.
Together, we can meet the same error bound with much few samples, and more importantly, with confidence.

\begin{figure}[tb]
\centering
\includegraphics[width=0.39\textwidth]{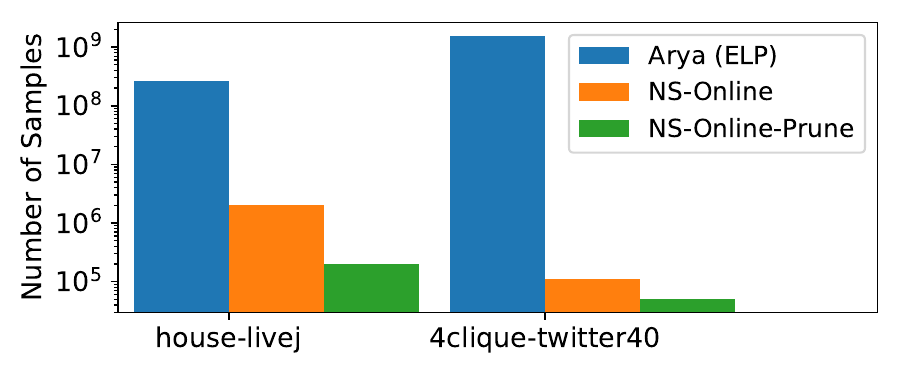}
\caption{The reduction on the number of samples.}
\label{fig:sample-reduction}
\end{figure}

\subsection{Prediction Accuracy of Cost Models} \label{subsect:model-accuracy}

\cref{fig:perf-predict-gs} shows the effectiveness of our cost model in predicting the running time of the \GS engine in \sys. 
As we increase the number of colors $c=\frac{1}{p}$ used in \GS, the total work is exponentially decreased. 
Thus, the \GS execution time rapidly becomes dominated by the preprocessing time spent on sparsifying the graph. 
For {\texttt{4clique} and \texttt{house} on \texttt{Lj} and \texttt{Fr}}, we see that the cost model precisely captures the exponential trend, as well as the stabilization window of the \GS engine that we discussed in \cref{fig:color-speed}. 
Note that when $c$ is extremely small, e.g., $c<5$, it is hard to accurately model the performance due to the steep slope in that curve.
However, in this case (which is a needle-in-the-hay case) we would favor the use of \NS or even exact counting, as \GS with a small $c$ won't be much faster than exact counting.

\begin{figure*}[tb]
\centering
\includegraphics[width=0.95\textwidth]{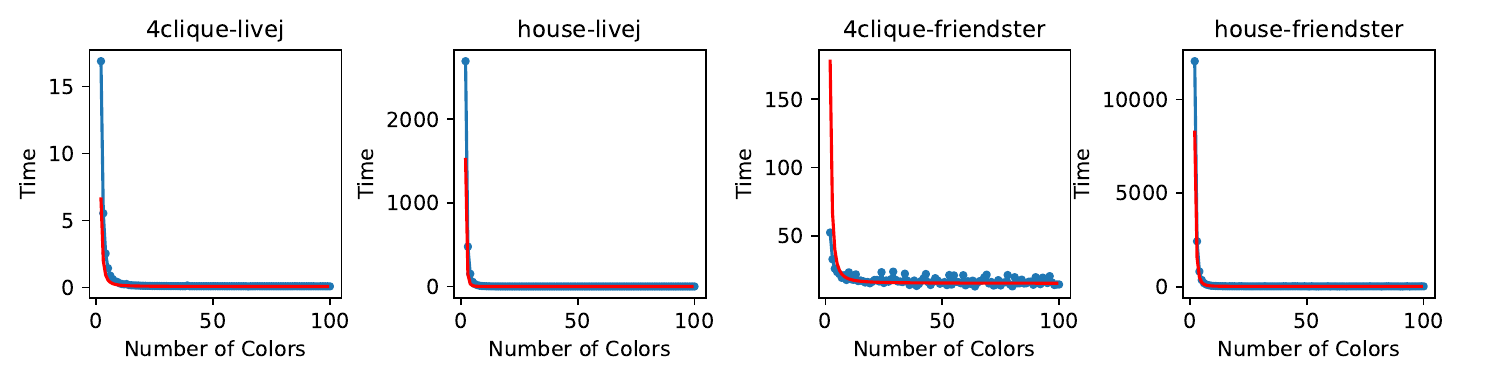}
\caption{The quality of performance prediction for \sys 's \GS Engine. The red lines are our predictions, while blue dots are actual time. }
\label{fig:perf-predict-gs}
\end{figure*} 

\begin{figure*}[tb]
\centering
\includegraphics[width=0.95\textwidth]{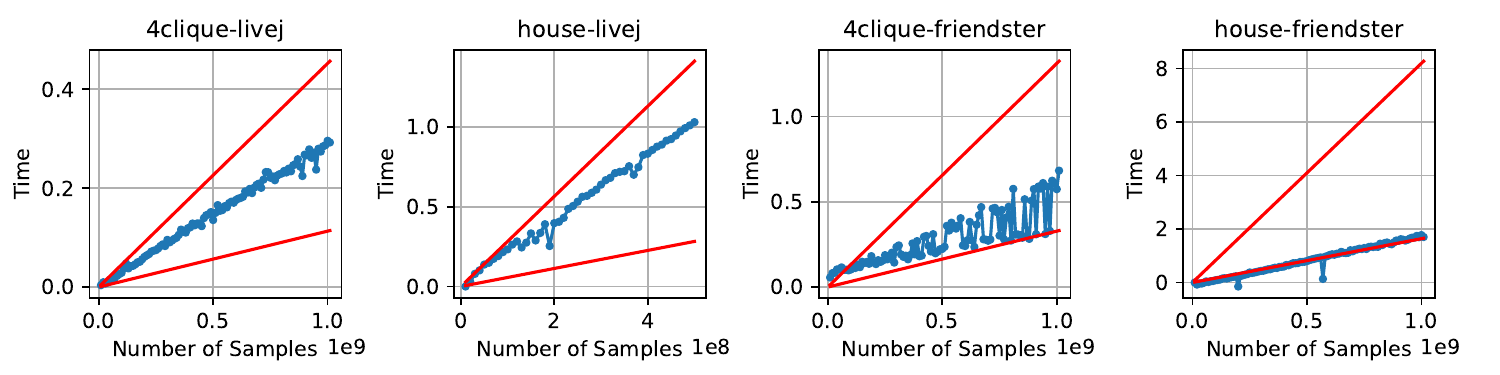}
\caption{The quality of performance prediction for \sys 's \NS Engine. The red lines are our predictions, while blue dots are actual time. 
}
\label{fig:perf-predict-ns}
\end{figure*}

\cref{fig:perf-predict-ns} shows the effectiveness of our \NS cost model. 
We see that our proposed \textit{performance cone} correctly captures the running time of the \NS engine at varying numbers of samples.
Note that for \texttt{Fr}, a relatively sparse graph, the execution time approaches the bottom of the performance cone as expected. 
We find in practice, that \GS and \NS are often predicted to have very distinct performance, so the width of the cone is not an obstacle in making the correct choice.
{The fluctuation in the sparse case 4\texttt{clique}-\texttt{Fr} is also expected, as sample hits are less frequent and thus more randomness is involved. }

\vspace{-6pt} 
\subsection{System Efficiency} \label{subsect:sys-efficiency}

\begin{figure}[tb]
\centering
\includegraphics[width=0.5\textwidth]{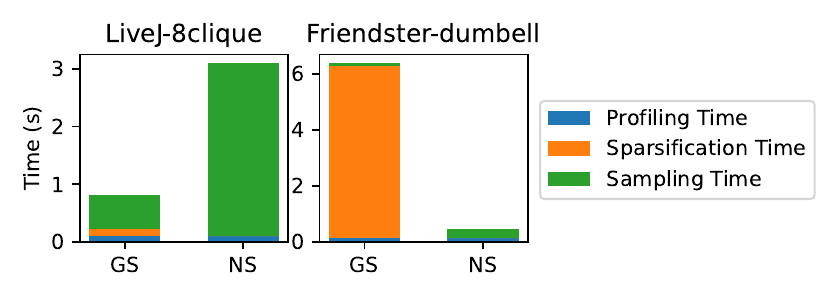}
\vspace{-13pt}
\caption{End-to-end time breakdown of \sys. }
\label{fig:breakdown}
\end{figure}

{\bf Timing Breakdown}. 
\cref{fig:breakdown} shows the breakdown of the execution time spent on different components of \sys.
We consider the profiling time, the \GS preprocessing time, and the sampling time (either exact counting in \GS or drawing samples in \NS).
We see that profiling remains a low percentage of the overall runtime.
As for the sampling time, \texttt{Lj}-8\texttt{clique} requires less time using the \GS engine, while \texttt{Fr}-\texttt{dumbbell} is processed faster using the \NS engine, which aligns with our cost model prediction and thus \sys-{\tt HY} makes the correct thresholding decision. 

\begin{figure}[tb]
\centering
\includegraphics[width=0.45\textwidth]{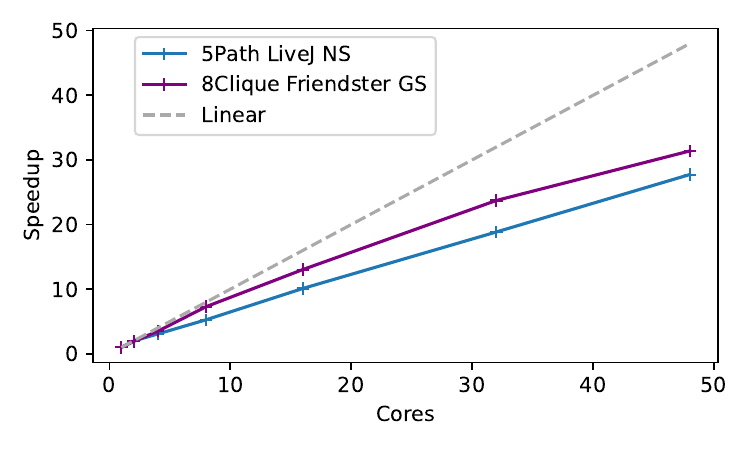}
\vspace{-7pt}
\caption{\sys speedup scaling over single-thread.}
\label{fig:scale}
\end{figure}

{\bf Scalability}. 
\cref{fig:scale} shows how the performance of \GS engine and \NS engine in \sys (error bound of 10\%) scales in response to the increase of parallel cores (i.e., the number of threads). 
We evaluate the \NS engine on \texttt{5path}-\texttt{Lj} which is a case that favors the use of the \NS engine (i.e., \NS is faster than \GS). 
\GS is evaluated on 8\texttt{clique}-\texttt{Fr} as 8-\texttt{clique} is a dense pattern and is rare in \texttt{Fr}, preferably executed in \GS. 
In both cases, we observe strong scaling that the execution time of both engines increases linearly as we range the number of cores from 1 to 48. 



\section{Related Work} \label{sect:related}



There exist many sampling methods other than neighbor sampling and graph sparsification. 
We discuss some of the typical methods in the following.
We do not compare with them in \cref{sect:eval} because they do not provide either generalization or automated termination (with confidence). 
However, they can potentially be used to replace the \GS engine in \sys. 
We leave it as a future work. 



\noindent
{\bf Color Coding}. 
It first colors each vertex in \dg using a color
randomly chosen from \{1,2,$\dots$,$c$\} ($c\geq k$), 
and then counts \textit{colorful} matches, i.e., every vertex in the matched subgraph has a unique color.
The requirement of distinct colors allows for heavy pruning: 
the number of colorful matches $Z$ can be naturally determined by a dynamic programming based counting routine~\cite{fascia}.
Color-coding is originally for finding paths or cycles~\cite{Alon,Color-Coding}, and is then adapted for motif counting~\cite{ParSE,CC-biomolecular,AE-CC}.
There also exist many parallel and distributed implementations \cite{fascia,CC-Parallel,SC-CC,SAHAD}.
In addition, the colorful matches can be further sampled~\cite{Motivo,Bressan1,Bressan}, 
instead of exactly counted, to reduce computation.
Like \GS, color-coding is also a coarse-grain scheme, which can be included in \sys.

{\bf Loop Perforation}. 
SampleMine~\cite{SampleMine} proposed to perforate the nested \texttt{for} loops in GPM programs, with a certain probability $p_i$ for the $i$-th loop. 
The count is then scaled based on the $p_i$.
SampleMine can be thought of as a generalization of the vanilla \NS scheme, as it collectively samples multiple candidates, instead of a single one, at a time.
Although it has larger sampling granularity than the vanilla \NS, 
its granularity is still limited by the egonet of a vertex or edge, similar to Egonet Sampling~\cite{EgonetMotif}. 
More importantly, SampleMine does not provide a systematic way for sampling termination, which is instead hand-tuned by executing the sampling procedure multiple times and manually observing convergence (i.e. small variance). 

\noindent
\textbf{Other Schemes and Approaches}.
Monte Carlo Markov Chain (MCMC) ~\cite{GUISE,TETRIS,GeneralRW,PSRW} defines a random walk over the set of subgraphs until it reaches stationarity. 
MCMC has been used for motif counting. 
However, it has been shown that MCMC can be extremely inefficient as the random walk may take a huge amount of steps to reach stationarity~\cite{Bressan,Bressan1,Motivo}.


\section{Conclusion}

Approximate \gpm (A-GPM) systems can be the backbone to support numerous real-world data analytics applications.
The key obstacles that prevent A-GPM systems from being adopted in practice is a) the lack of stable, confident termination mechanism and b) poor performance and scalability when dealing with the ``hard'' cases.
We present \sys, an accurate, high performance and scalable A-GPM systems. 
\sys involves two key innovations that remove the obstacles.
First, we propose a novel online convergence detection mechanism, which can provide theoretical guarantee on prediction confidence and also yield stable termination condition.
Second, we introduce pruning techniques into sampling to improve its success rate, 
and propose a hybrid method to adaptively select the best-performing sampling scheme from two complementary schemes, based on our proposed cost models.
The resulting system, \sys, {achieves an average of \usoveraryaall speedup} over the state-of-the-art A-GPM systems,
and manages to rapidly mine billion-scale graphs.

\bibliographystyle{ACM-Reference-Format}
\bibliography{refs}


\begin{thebibliography}{70}


\ifx \showCODEN    \undefined \def \showCODEN     #1{\unskip}     \fi
\ifx \showDOI      \undefined \def \showDOI       #1{#1}\fi
\ifx \showISBNx    \undefined \def \showISBNx     #1{\unskip}     \fi
\ifx \showISBNxiii \undefined \def \showISBNxiii  #1{\unskip}     \fi
\ifx \showISSN     \undefined \def \showISSN      #1{\unskip}     \fi
\ifx \showLCCN     \undefined \def \showLCCN      #1{\unskip}     \fi
\ifx \shownote     \undefined \def \shownote      #1{#1}          \fi
\ifx \showarticletitle \undefined \def \showarticletitle #1{#1}   \fi
\ifx \showURL      \undefined \def \showURL       {\relax}        \fi
\providecommand\bibfield[2]{#2}
\providecommand\bibinfo[2]{#2}
\providecommand\natexlab[1]{#1}
\providecommand\showeprint[2][]{arXiv:#2}

\bibitem[Abdelhamid et~al\mbox{.}(2016)]%
        {Scalemine}
\bibfield{author}{\bibinfo{person}{Ehab Abdelhamid}, \bibinfo{person}{Ibrahim
  Abdelaziz}, \bibinfo{person}{Panos Kalnis}, \bibinfo{person}{Zuhair Khayyat},
  {and} \bibinfo{person}{Fuad Jamour}.} \bibinfo{year}{2016}\natexlab{}.
\newblock \showarticletitle{Scalemine: Scalable Parallel Frequent Subgraph
  Mining in a Single Large Graph}. In \bibinfo{booktitle}{\emph{Proceedings of
  the International Conference for High Performance Computing, Networking,
  Storage and Analysis}} (Salt Lake City, Utah) \emph{(\bibinfo{series}{SC
  '16})}. \bibinfo{publisher}{IEEE Press}, \bibinfo{address}{Piscataway, NJ,
  USA}, Article \bibinfo{articleno}{61}, \bibinfo{numpages}{12}~pages.
\newblock
\showISBNx{978-1-4673-8815-3}
\urldef\tempurl%
\url{http://dl.acm.org/citation.cfm?id=3014904.3014986}
\showURL{%
\tempurl}


\bibitem[Ahmed et~al\mbox{.}(2014)]%
        {SampleHold}
\bibfield{author}{\bibinfo{person}{Nesreen~K. Ahmed}, \bibinfo{person}{Nick
  Duffield}, \bibinfo{person}{Jennifer Neville}, {and} \bibinfo{person}{Ramana
  Kompella}.} \bibinfo{year}{2014}\natexlab{}.
\newblock \showarticletitle{Graph Sample and Hold: A Framework for Big-Graph
  Analytics}. In \bibinfo{booktitle}{\emph{Proceedings of the 20th ACM SIGKDD
  International Conference on Knowledge Discovery and Data Mining}} (New York,
  New York, USA) \emph{(\bibinfo{series}{KDD '14})}.
  \bibinfo{publisher}{Association for Computing Machinery},
  \bibinfo{address}{New York, NY, USA}, \bibinfo{pages}{1446–1455}.
\newblock
\showISBNx{9781450329569}
\urldef\tempurl%
\url{https://doi.org/10.1145/2623330.2623757}
\showDOI{\tempurl}


\bibitem[Aliakbarpour et~al\mbox{.}(2018)]%
        {StarCounting}
\bibfield{author}{\bibinfo{person}{Maryam Aliakbarpour},
  \bibinfo{person}{Amartya~Shankha Biswas}, \bibinfo{person}{Themis Gouleakis},
  \bibinfo{person}{John Peebles}, \bibinfo{person}{Ronitt Rubinfeld}, {and}
  \bibinfo{person}{Anak Yodpinyanee}.} \bibinfo{year}{2018}\natexlab{}.
\newblock \showarticletitle{Sublinear-Time Algorithms for Counting Star
  Subgraphs via Edge Sampling}.
\newblock \bibinfo{journal}{\emph{Algorithmica}} \bibinfo{volume}{80},
  \bibinfo{number}{2} (\bibinfo{date}{feb} \bibinfo{year}{2018}),
  \bibinfo{pages}{668–697}.
\newblock
\showISSN{0178-4617}
\urldef\tempurl%
\url{https://doi.org/10.1007/s00453-017-0287-3}
\showDOI{\tempurl}


\bibitem[Alon et~al\mbox{.}(2008)]%
        {CC-biomolecular}
\bibfield{author}{\bibinfo{person}{Noga Alon}, \bibinfo{person}{Phuong Dao},
  \bibinfo{person}{Iman Hajirasouliha}, \bibinfo{person}{Fereydoun
  Hormozdiari}, {and} \bibinfo{person}{S~Cenk Sahinalp}.}
  \bibinfo{year}{2008}\natexlab{}.
\newblock \showarticletitle{Biomolecular network motif counting and discovery
  by color coding}.
\newblock \bibinfo{journal}{\emph{Bioinformatics}} \bibinfo{volume}{24},
  \bibinfo{number}{13} (\bibinfo{year}{2008}), \bibinfo{pages}{i241--i249}.
\newblock


\bibitem[Alon et~al\mbox{.}(1994)]%
        {Alon}
\bibfield{author}{\bibinfo{person}{Noga Alon}, \bibinfo{person}{Raphy Yuster},
  {and} \bibinfo{person}{Uri Zwick}.} \bibinfo{year}{1994}\natexlab{}.
\newblock \showarticletitle{Color-coding: A New Method for Finding Simple
  Paths, Cycles and Other Small Subgraphs Within Large Graphs}. In
  \bibinfo{booktitle}{\emph{Proceedings of the Twenty-sixth Annual ACM
  Symposium on Theory of Computing}} (Montreal, Quebec, Canada)
  \emph{(\bibinfo{series}{STOC '94})}. \bibinfo{publisher}{ACM},
  \bibinfo{address}{New York, NY, USA}, \bibinfo{pages}{326--335}.
\newblock
\showISBNx{0-89791-663-8}
\urldef\tempurl%
\url{https://doi.org/10.1145/195058.195179}
\showDOI{\tempurl}


\bibitem[Alon et~al\mbox{.}(1995)]%
        {Color-Coding}
\bibfield{author}{\bibinfo{person}{Noga Alon}, \bibinfo{person}{Raphael
  Yuster}, {and} \bibinfo{person}{Uri Zwick}.} \bibinfo{year}{1995}\natexlab{}.
\newblock \showarticletitle{Color-Coding}.
\newblock \bibinfo{journal}{\emph{J. ACM}} \bibinfo{volume}{42},
  \bibinfo{number}{4} (\bibinfo{date}{jul} \bibinfo{year}{1995}),
  \bibinfo{pages}{844–856}.
\newblock
\showISSN{0004-5411}
\urldef\tempurl%
\url{https://doi.org/10.1145/210332.210337}
\showDOI{\tempurl}


\bibitem[Assadi et~al\mbox{.}(2018)]%
        {EdgeSampling}
\bibfield{author}{\bibinfo{person}{Sepehr Assadi}, \bibinfo{person}{Mikhail
  Kapralov}, {and} \bibinfo{person}{Sanjeev Khanna}.}
  \bibinfo{year}{2018}\natexlab{}.
\newblock \showarticletitle{A Simple Sublinear-Time Algorithm for Counting
  Arbitrary Subgraphs via Edge Sampling}. In
  \bibinfo{booktitle}{\emph{Information Technology Convergence and Services}}.
\newblock


\bibitem[Bar-Yossef et~al\mbox{.}(2002)]%
        {StreamingTriangleCounting}
\bibfield{author}{\bibinfo{person}{Ziv Bar-Yossef}, \bibinfo{person}{Ravi
  Kumar}, {and} \bibinfo{person}{D. Sivakumar}.}
  \bibinfo{year}{2002}\natexlab{}.
\newblock \showarticletitle{Reductions in Streaming Algorithms, with an
  Application to Counting Triangles in Graphs}. In
  \bibinfo{booktitle}{\emph{Proceedings of the Thirteenth Annual ACM-SIAM
  Symposium on Discrete Algorithms}} (San Francisco, California)
  \emph{(\bibinfo{series}{SODA '02})}. \bibinfo{publisher}{Society for
  Industrial and Applied Mathematics}, \bibinfo{address}{USA},
  \bibinfo{pages}{623–632}.
\newblock
\showISBNx{089871513X}


\bibitem[Bencz\'{u}r and Karger(1996)]%
        {ApproxMinCut}
\bibfield{author}{\bibinfo{person}{Andr\'{a}s~A. Bencz\'{u}r} {and}
  \bibinfo{person}{David~R. Karger}.} \bibinfo{year}{1996}\natexlab{}.
\newblock \showarticletitle{Approximating S-t Minimum Cuts in \~{O}(N2) Time}.
  In \bibinfo{booktitle}{\emph{Proceedings of the Twenty-Eighth Annual ACM
  Symposium on Theory of Computing}} (Philadelphia, Pennsylvania, USA)
  \emph{(\bibinfo{series}{STOC '96})}. \bibinfo{publisher}{Association for
  Computing Machinery}, \bibinfo{address}{New York, NY, USA},
  \bibinfo{pages}{47–55}.
\newblock
\showISBNx{0897917855}
\urldef\tempurl%
\url{https://doi.org/10.1145/237814.237827}
\showDOI{\tempurl}


\bibitem[Bera and Seshadhri(2020)]%
        {TETRIS}
\bibfield{author}{\bibinfo{person}{Suman~K. Bera} {and} \bibinfo{person}{C.
  Seshadhri}.} \bibinfo{year}{2020}\natexlab{}.
\newblock \showarticletitle{How to Count Triangles, without Seeing the Whole
  Graph}. In \bibinfo{booktitle}{\emph{Proceedings of the 26th ACM SIGKDD
  International Conference on Knowledge Discovery \& Data Mining}} (Virtual
  Event, CA, USA) \emph{(\bibinfo{series}{KDD '20})}.
  \bibinfo{publisher}{Association for Computing Machinery},
  \bibinfo{address}{New York, NY, USA}, \bibinfo{pages}{306–316}.
\newblock
\showISBNx{9781450379984}
\urldef\tempurl%
\url{https://doi.org/10.1145/3394486.3403073}
\showDOI{\tempurl}


\bibitem[Bhatia and Rani(2018)]%
        {Ap-FSM}
\bibfield{author}{\bibinfo{person}{Vandana Bhatia} {and}
  \bibinfo{person}{Rinkle Rani}.} \bibinfo{year}{2018}\natexlab{}.
\newblock \showarticletitle{Ap-FSM: A parallel algorithm for approximate
  frequent subgraph mining using Pregel}.
\newblock \bibinfo{journal}{\emph{Expert Systems with Applications}}
  \bibinfo{volume}{106} (\bibinfo{year}{2018}), \bibinfo{pages}{217--232}.
\newblock
\showISSN{0957-4174}
\urldef\tempurl%
\url{https://doi.org/10.1016/j.eswa.2018.04.010}
\showDOI{\tempurl}


\bibitem[Bhuiyan et~al\mbox{.}(2012)]%
        {GUISE}
\bibfield{author}{\bibinfo{person}{Mansurul~A. Bhuiyan},
  \bibinfo{person}{Mahmudur Rahman}, \bibinfo{person}{Mahmuda Rahman}, {and}
  \bibinfo{person}{Mohammad Al~Hasan}.} \bibinfo{year}{2012}\natexlab{}.
\newblock \showarticletitle{GUISE: Uniform Sampling of Graphlets for Large
  Graph Analysis}. In \bibinfo{booktitle}{\emph{2012 IEEE 12th International
  Conference on Data Mining}}. \bibinfo{pages}{91--100}.
\newblock
\urldef\tempurl%
\url{https://doi.org/10.1109/ICDM.2012.87}
\showDOI{\tempurl}


\bibitem[Boldi et~al\mbox{.}(2008)]%
        {uk2007}
\bibfield{author}{\bibinfo{person}{Paolo Boldi}, \bibinfo{person}{Massimo
  Santini}, {and} \bibinfo{person}{Sebastiano Vigna}.}
  \bibinfo{year}{2008}\natexlab{}.
\newblock \showarticletitle{A Large Time-Aware Graph}.
\newblock \bibinfo{journal}{\emph{SIGIR Forum}} \bibinfo{volume}{42},
  \bibinfo{number}{2} (\bibinfo{year}{2008}), \bibinfo{pages}{33--38}.
\newblock


\bibitem[Boldi and Vigna(2004)]%
        {gsh15}
\bibfield{author}{\bibinfo{person}{Paolo Boldi} {and}
  \bibinfo{person}{Sebastiano Vigna}.} \bibinfo{year}{2004}\natexlab{}.
\newblock \showarticletitle{{The WebGraph Framework I: Compression
  Techniques}}. In \bibinfo{booktitle}{\emph{Proceedings of the 13th
  International Conference on World Wide Web}} (New York, NY, USA)
  \emph{(\bibinfo{series}{WWW '04})}. \bibinfo{publisher}{ACM},
  \bibinfo{address}{New York, NY, USA}, \bibinfo{pages}{595--602}.
\newblock
\showISBNx{1-58113-844-X}
\urldef\tempurl%
\url{https://doi.org/10.1145/988672.988752}
\showDOI{\tempurl}


\bibitem[Bressan(2021)]%
        {Near-Optimal}
\bibfield{author}{\bibinfo{person}{Marco Bressan}.}
  \bibinfo{year}{2021}\natexlab{}.
\newblock \showarticletitle{Efficient and Near-Optimal Algorithms for Sampling
  Connected Subgraphs}. In \bibinfo{booktitle}{\emph{Proceedings of the 53rd
  Annual ACM SIGACT Symposium on Theory of Computing}} (Virtual, Italy)
  \emph{(\bibinfo{series}{STOC 2021})}. \bibinfo{publisher}{Association for
  Computing Machinery}, \bibinfo{address}{New York, NY, USA},
  \bibinfo{pages}{1132–1143}.
\newblock
\showISBNx{9781450380539}
\urldef\tempurl%
\url{https://doi.org/10.1145/3406325.3451042}
\showDOI{\tempurl}


\bibitem[Bressan et~al\mbox{.}(2017)]%
        {Bressan1}
\bibfield{author}{\bibinfo{person}{Marco Bressan}, \bibinfo{person}{Flavio
  Chierichetti}, \bibinfo{person}{Ravi Kumar}, \bibinfo{person}{Stefano
  Leucci}, {and} \bibinfo{person}{Alessandro Panconesi}.}
  \bibinfo{year}{2017}\natexlab{}.
\newblock \showarticletitle{Counting Graphlets: Space vs Time}. In
  \bibinfo{booktitle}{\emph{Proceedings of the Tenth ACM International
  Conference on Web Search and Data Mining}} (Cambridge, United Kingdom)
  \emph{(\bibinfo{series}{WSDM '17})}. \bibinfo{publisher}{ACM},
  \bibinfo{address}{New York, NY, USA}, \bibinfo{pages}{557--566}.
\newblock
\showISBNx{978-1-4503-4675-7}
\urldef\tempurl%
\url{https://doi.org/10.1145/3018661.3018732}
\showDOI{\tempurl}


\bibitem[Bressan et~al\mbox{.}(2018)]%
        {Bressan}
\bibfield{author}{\bibinfo{person}{Marco Bressan}, \bibinfo{person}{Flavio
  Chierichetti}, \bibinfo{person}{Ravi Kumar}, \bibinfo{person}{Stefano
  Leucci}, {and} \bibinfo{person}{Alessandro Panconesi}.}
  \bibinfo{year}{2018}\natexlab{}.
\newblock \showarticletitle{Motif Counting Beyond Five Nodes}.
\newblock \bibinfo{journal}{\emph{ACM Trans. Knowl. Discov. Data}}
  \bibinfo{volume}{12}, \bibinfo{number}{4}, Article \bibinfo{articleno}{48}
  (\bibinfo{date}{April} \bibinfo{year}{2018}), \bibinfo{numpages}{25}~pages.
\newblock
\showISSN{1556-4681}
\urldef\tempurl%
\url{https://doi.org/10.1145/3186586}
\showDOI{\tempurl}


\bibitem[Bressan et~al\mbox{.}(2019)]%
        {Motivo}
\bibfield{author}{\bibinfo{person}{Marco Bressan}, \bibinfo{person}{Stefano
  Leucci}, {and} \bibinfo{person}{Alessandro Panconesi}.}
  \bibinfo{year}{2019}\natexlab{}.
\newblock \showarticletitle{Motivo: Fast Motif Counting via Succinct Color
  Coding and Adaptive Sampling}.
\newblock \bibinfo{journal}{\emph{Proc. VLDB Endow.}} \bibinfo{volume}{12},
  \bibinfo{number}{11} (\bibinfo{date}{July} \bibinfo{year}{2019}),
  \bibinfo{pages}{1651–1663}.
\newblock
\showISSN{2150-8097}
\urldef\tempurl%
\url{https://doi.org/10.14778/3342263.3342640}
\showDOI{\tempurl}


\bibitem[Buriol et~al\mbox{.}(2006)]%
        {Buriol}
\bibfield{author}{\bibinfo{person}{Luciana~S. Buriol}, \bibinfo{person}{Gereon
  Frahling}, \bibinfo{person}{Stefano Leonardi}, \bibinfo{person}{Alberto
  Marchetti-Spaccamela}, {and} \bibinfo{person}{Christian Sohler}.}
  \bibinfo{year}{2006}\natexlab{}.
\newblock \showarticletitle{Counting Triangles in Data Streams}. In
  \bibinfo{booktitle}{\emph{Proceedings of the Twenty-Fifth ACM
  SIGMOD-SIGACT-SIGART Symposium on Principles of Database Systems}} (Chicago,
  IL, USA) \emph{(\bibinfo{series}{PODS '06})}. \bibinfo{publisher}{Association
  for Computing Machinery}, \bibinfo{address}{New York, NY, USA},
  \bibinfo{pages}{253–262}.
\newblock
\showISBNx{1595933182}
\urldef\tempurl%
\url{https://doi.org/10.1145/1142351.1142388}
\showDOI{\tempurl}


\bibitem[Chakaravarthy et~al\mbox{.}(2016)]%
        {SC-CC}
\bibfield{author}{\bibinfo{person}{V.~T. Chakaravarthy}, \bibinfo{person}{M.
  Kapralov}, \bibinfo{person}{P. Murali}, \bibinfo{person}{F. Petrini},
  \bibinfo{person}{X. Que}, \bibinfo{person}{Y. Sabharwal}, {and}
  \bibinfo{person}{B. Schieber}.} \bibinfo{year}{2016}\natexlab{}.
\newblock \showarticletitle{Subgraph Counting: Color Coding Beyond Trees}. In
  \bibinfo{booktitle}{\emph{2016 IEEE International Parallel and Distributed
  Processing Symposium (IPDPS)}}. \bibinfo{publisher}{IEEE Computer Society},
  \bibinfo{address}{Los Alamitos, CA, USA}, \bibinfo{pages}{2--11}.
\newblock
\showISSN{1530-2075}
\urldef\tempurl%
\url{https://doi.org/10.1109/IPDPS.2016.122}
\showDOI{\tempurl}


\bibitem[Chen et~al\mbox{.}({[n.\,d.]})]%
        {heavyEdge}
\bibfield{author}{\bibinfo{person}{J. Chen}, \bibinfo{person}{T. Eden},
  \bibinfo{person}{P. Indyk}, \bibinfo{person}{S. Narayanan},
  \bibinfo{person}{R. Rubinfeld}, \bibinfo{person}{S. Silwal},
  \bibinfo{person}{D. Woodruff}, {and} \bibinfo{person}{M. Zhang}.}
  \bibinfo{year}{[n.\,d.]}\natexlab{}.
\newblock \showarticletitle{Triangle and Four Cycle Counting with Predictions
  in Graph Streams}.
\newblock \bibinfo{journal}{\emph{Tenth International Conference on Learning
  Representations (ICLR 2022)}} (\bibinfo{year}{[n.\,d.]}).
\newblock
\urldef\tempurl%
\url{https://par.nsf.gov/biblio/10338743}
\showURL{%
\tempurl}


\bibitem[Chen and Qian(2022)]%
        {DecoMine}
\bibfield{author}{\bibinfo{person}{Jingji Chen} {and} \bibinfo{person}{Xuehai
  Qian}.} \bibinfo{year}{2022}\natexlab{}.
\newblock \showarticletitle{DecoMine: A Compilation-Based Graph Pattern Mining
  System with Pattern Decomposition}. In \bibinfo{booktitle}{\emph{Proceedings
  of the 28th ACM International Conference on Architectural Support for
  Programming Languages and Operating Systems, Volume 1}} (Vancouver, BC,
  Canada) \emph{(\bibinfo{series}{ASPLOS 2023})}.
  \bibinfo{publisher}{Association for Computing Machinery},
  \bibinfo{address}{New York, NY, USA}, \bibinfo{pages}{47–61}.
\newblock
\showISBNx{9781450399159}
\urldef\tempurl%
\url{https://doi.org/10.1145/3567955.3567956}
\showDOI{\tempurl}


\bibitem[Chen et~al\mbox{.}(2022)]%
        {TC-FourCycle}
\bibfield{author}{\bibinfo{person}{Justin~Y Chen}, \bibinfo{person}{Talya
  Eden}, \bibinfo{person}{Piotr Indyk}, \bibinfo{person}{Honghao Lin},
  \bibinfo{person}{Shyam Narayanan}, \bibinfo{person}{Ronitt Rubinfeld},
  \bibinfo{person}{Sandeep Silwal}, \bibinfo{person}{Tal Wagner},
  \bibinfo{person}{David Woodruff}, {and} \bibinfo{person}{Michael Zhang}.}
  \bibinfo{year}{2022}\natexlab{}.
\newblock \showarticletitle{Triangle and Four Cycle Counting with Predictions
  in Graph Streams}. In \bibinfo{booktitle}{\emph{International Conference on
  Learning Representations}}.
\newblock
\urldef\tempurl%
\url{https://openreview.net/forum?id=8in_5gN9I0}
\showURL{%
\tempurl}


\bibitem[Chen et~al\mbox{.}(2021)]%
        {Sandslash}
\bibfield{author}{\bibinfo{person}{Xuhao Chen}, \bibinfo{person}{Roshan
  Dathathri}, \bibinfo{person}{Gurbinder Gill}, \bibinfo{person}{Loc Hoang},
  {and} \bibinfo{person}{Keshav Pingali}.} \bibinfo{year}{2021}\natexlab{}.
\newblock \showarticletitle{{Sandslash: A Two-Level Framework for Efficient
  Graph Pattern Mining}}. In \bibinfo{booktitle}{\emph{Proceedings of the 35th
  ACM International Conference on Supercomputing}} \emph{(\bibinfo{series}{ICS
  '21})}. \bibinfo{numpages}{14}~pages.
\newblock


\bibitem[Chen et~al\mbox{.}(2020)]%
        {Pangolin}
\bibfield{author}{\bibinfo{person}{Xuhao Chen}, \bibinfo{person}{Roshan
  Dathathri}, \bibinfo{person}{Gurbinder Gill}, {and} \bibinfo{person}{Keshav
  Pingali}.} \bibinfo{year}{2020}\natexlab{}.
\newblock \showarticletitle{Pangolin: An Efficient and Flexible Graph Mining
  System on CPU and GPU}.
\newblock \bibinfo{journal}{\emph{Proc. VLDB Endow.}} \bibinfo{volume}{13},
  \bibinfo{number}{8} (\bibinfo{date}{Aug.} \bibinfo{year}{2020}).
\newblock
\showISSN{2150-8097}
\urldef\tempurl%
\url{https://doi.org/10.14778/3389133.3389137}
\showDOI{\tempurl}


\bibitem[Chen et~al\mbox{.}(2016)]%
        {GeneralRW}
\bibfield{author}{\bibinfo{person}{Xiaowei Chen}, \bibinfo{person}{Yongkun Li},
  \bibinfo{person}{Pinghui Wang}, {and} \bibinfo{person}{John C.~S. Lui}.}
  \bibinfo{year}{2016}\natexlab{}.
\newblock \showarticletitle{A General Framework for Estimating Graphlet
  Statistics via Random Walk}.
\newblock \bibinfo{journal}{\emph{Proc. VLDB Endow.}} \bibinfo{volume}{10},
  \bibinfo{number}{3} (\bibinfo{date}{nov} \bibinfo{year}{2016}),
  \bibinfo{pages}{253–264}.
\newblock
\showISSN{2150-8097}
\urldef\tempurl%
\url{https://doi.org/10.14778/3021924.3021940}
\showDOI{\tempurl}


\bibitem[Dias et~al\mbox{.}(2019)]%
        {Fractal}
\bibfield{author}{\bibinfo{person}{Vinicius Dias}, \bibinfo{person}{Carlos
  H.~C. Teixeira}, \bibinfo{person}{Dorgival Guedes}, \bibinfo{person}{Wagner
  Meira}, {and} \bibinfo{person}{Srinivasan Parthasarathy}.}
  \bibinfo{year}{2019}\natexlab{}.
\newblock \showarticletitle{Fractal: A General-Purpose Graph Pattern Mining
  System}. In \bibinfo{booktitle}{\emph{Proceedings of the 2019 International
  Conference on Management of Data}} (Amsterdam, Netherlands)
  \emph{(\bibinfo{series}{SIGMOD '19})}. \bibinfo{publisher}{ACM},
  \bibinfo{address}{New York, NY, USA}, \bibinfo{pages}{1357--1374}.
\newblock
\showISBNx{978-1-4503-5643-5}
\urldef\tempurl%
\url{https://doi.org/10.1145/3299869.3319875}
\showDOI{\tempurl}


\bibitem[Elenberg et~al\mbox{.}(2015)]%
        {BeyondTri}
\bibfield{author}{\bibinfo{person}{Ethan~R. Elenberg},
  \bibinfo{person}{Karthikeyan Shanmugam}, \bibinfo{person}{Michael
  Borokhovich}, {and} \bibinfo{person}{Alexandros~G. Dimakis}.}
  \bibinfo{year}{2015}\natexlab{}.
\newblock \showarticletitle{Beyond Triangles: A Distributed Framework for
  Estimating 3-profiles of Large Graphs}. In
  \bibinfo{booktitle}{\emph{Proceedings of the 21th ACM SIGKDD International
  Conference on Knowledge Discovery and Data Mining}} (Sydney, NSW, Australia)
  \emph{(\bibinfo{series}{KDD '15})}. \bibinfo{publisher}{ACM},
  \bibinfo{address}{New York, NY, USA}, \bibinfo{pages}{229--238}.
\newblock
\showISBNx{978-1-4503-3664-2}
\urldef\tempurl%
\url{https://doi.org/10.1145/2783258.2783413}
\showDOI{\tempurl}


\bibitem[Fung et~al\mbox{.}(2011)]%
        {GFforGraphSpar}
\bibfield{author}{\bibinfo{person}{Wai~Shing Fung}, \bibinfo{person}{Ramesh
  Hariharan}, \bibinfo{person}{Nicholas~J.A. Harvey}, {and}
  \bibinfo{person}{Debmalya Panigrahi}.} \bibinfo{year}{2011}\natexlab{}.
\newblock \showarticletitle{A General Framework for Graph Sparsification}. In
  \bibinfo{booktitle}{\emph{Proceedings of the Forty-Third Annual ACM Symposium
  on Theory of Computing}} (San Jose, California, USA)
  \emph{(\bibinfo{series}{STOC '11})}. \bibinfo{publisher}{Association for
  Computing Machinery}, \bibinfo{address}{New York, NY, USA},
  \bibinfo{pages}{71–80}.
\newblock
\showISBNx{9781450306911}
\urldef\tempurl%
\url{https://doi.org/10.1145/1993636.1993647}
\showDOI{\tempurl}


\bibitem[H{\"u}ffner et~al\mbox{.}(2008)]%
        {AE-CC}
\bibfield{author}{\bibinfo{person}{Falk H{\"u}ffner},
  \bibinfo{person}{Sebastian Wernicke}, {and} \bibinfo{person}{Thomas
  Zichner}.} \bibinfo{year}{2008}\natexlab{}.
\newblock \showarticletitle{Algorithm engineering for color-coding with
  applications to signaling pathway detection}.
\newblock \bibinfo{journal}{\emph{Algorithmica}}  \bibinfo{volume}{52}
  (\bibinfo{year}{2008}), \bibinfo{pages}{114--132}.
\newblock


\bibitem[Iyer et~al\mbox{.}(2018)]%
        {ASAP}
\bibfield{author}{\bibinfo{person}{Anand~Padmanabha Iyer},
  \bibinfo{person}{Zaoxing Liu}, \bibinfo{person}{Xin Jin},
  \bibinfo{person}{Shivaram Venkataraman}, \bibinfo{person}{Vladimir
  Braverman}, {and} \bibinfo{person}{Ion Stoica}.}
  \bibinfo{year}{2018}\natexlab{}.
\newblock \showarticletitle{ASAP: Fast, Approximate Graph Pattern Mining at
  Scale}. In \bibinfo{booktitle}{\emph{Proceedings of the 12th USENIX
  Conference on Operating Systems Design and Implementation}} (Carlsbad, CA,
  USA) \emph{(\bibinfo{series}{OSDI'18})}. \bibinfo{publisher}{USENIX
  Association}, \bibinfo{address}{Berkeley, CA, USA},
  \bibinfo{pages}{745--761}.
\newblock
\showISBNx{978-1-931971-47-8}
\urldef\tempurl%
\url{http://dl.acm.org/citation.cfm?id=3291168.3291224}
\showURL{%
\tempurl}


\bibitem[Jamshidi et~al\mbox{.}(2020)]%
        {Peregrine}
\bibfield{author}{\bibinfo{person}{Kasra Jamshidi}, \bibinfo{person}{Rakesh
  Mahadasa}, {and} \bibinfo{person}{Keval Vora}.}
  \bibinfo{year}{2020}\natexlab{}.
\newblock \showarticletitle{Peregrine: A Pattern-Aware Graph Mining System}. In
  \bibinfo{booktitle}{\emph{Proceedings of the Fifteenth EuroSys Conference}}
  \emph{(\bibinfo{series}{EuroSys ’20})}.
\newblock


\bibitem[Jha et~al\mbox{.}(2015)]%
        {PathSampling}
\bibfield{author}{\bibinfo{person}{Madhav Jha}, \bibinfo{person}{C. Seshadhri},
  {and} \bibinfo{person}{Ali Pinar}.} \bibinfo{year}{2015}\natexlab{}.
\newblock \showarticletitle{Path Sampling: A Fast and Provable Method for
  Estimating 4-Vertex Subgraph Counts}. In
  \bibinfo{booktitle}{\emph{Proceedings of the 24th International Conference on
  World Wide Web}} (Florence, Italy) \emph{(\bibinfo{series}{WWW '15})}.
  \bibinfo{publisher}{International World Wide Web Conferences Steering
  Committee}, \bibinfo{address}{Republic and Canton of Geneva, Switzerland},
  \bibinfo{pages}{495--505}.
\newblock
\showISBNx{978-1-4503-3469-3}
\urldef\tempurl%
\url{https://doi.org/10.1145/2736277.2741101}
\showDOI{\tempurl}


\bibitem[Jiang et~al\mbox{.}(2023)]%
        {SampleMine}
\bibfield{author}{\bibinfo{person}{Peng Jiang}, \bibinfo{person}{Yihua Wei},
  \bibinfo{person}{Jiya Su}, \bibinfo{person}{Rujia Wang}, {and}
  \bibinfo{person}{Bo Wu}.} \bibinfo{year}{2023}\natexlab{}.
\newblock \showarticletitle{SampleMine: A Framework for Applying Random
  Sampling to Subgraph Pattern Mining through Loop Perforation}. In
  \bibinfo{booktitle}{\emph{Proceedings of the International Conference on
  Parallel Architectures and Compilation Techniques}} (Chicago, Illinois)
  \emph{(\bibinfo{series}{PACT '22})}. \bibinfo{publisher}{Association for
  Computing Machinery}, \bibinfo{address}{New York, NY, USA},
  \bibinfo{pages}{185–197}.
\newblock
\showISBNx{9781450398688}
\urldef\tempurl%
\url{https://doi.org/10.1145/3559009.3569658}
\showDOI{\tempurl}


\bibitem[Kallaugher et~al\mbox{.}(2019)]%
        {StreamingCount}
\bibfield{author}{\bibinfo{person}{John Kallaugher}, \bibinfo{person}{Andrew
  McGregor}, \bibinfo{person}{Eric Price}, {and} \bibinfo{person}{Sofya
  Vorotnikova}.} \bibinfo{year}{2019}\natexlab{}.
\newblock \showarticletitle{The Complexity of Counting Cycles in the Adjacency
  List Streaming Model}. In \bibinfo{booktitle}{\emph{Proceedings of the 38th
  ACM SIGMOD-SIGACT-SIGAI Symposium on Principles of Database Systems}}
  (Amsterdam, Netherlands) \emph{(\bibinfo{series}{PODS '19})}.
  \bibinfo{publisher}{Association for Computing Machinery},
  \bibinfo{address}{New York, NY, USA}, \bibinfo{pages}{119–133}.
\newblock
\showISBNx{9781450362276}
\urldef\tempurl%
\url{https://doi.org/10.1145/3294052.3319706}
\showDOI{\tempurl}


\bibitem[Kuramochi and Karypis(2004)]%
        {GREW}
\bibfield{author}{\bibinfo{person}{Michihiro Kuramochi} {and}
  \bibinfo{person}{George Karypis}.} \bibinfo{year}{2004}\natexlab{}.
\newblock \showarticletitle{GREW-A Scalable Frequent Subgraph Discovery
  Algorithm}. In \bibinfo{booktitle}{\emph{Proceedings of the Fourth IEEE
  International Conference on Data Mining}} \emph{(\bibinfo{series}{ICDM
  '04})}. \bibinfo{publisher}{IEEE Computer Society}, \bibinfo{address}{USA},
  \bibinfo{pages}{439–442}.
\newblock
\showISBNx{0769521428}


\bibitem[Kwak et~al\mbox{.}(2010)]%
        {twitter40}
\bibfield{author}{\bibinfo{person}{Haewoon Kwak}, \bibinfo{person}{Changhyun
  Lee}, \bibinfo{person}{Hosung Park}, {and} \bibinfo{person}{Sue Moon}.}
  \bibinfo{year}{2010}\natexlab{}.
\newblock \showarticletitle{What is Twitter, a Social Network or a News
  Media?}. In \bibinfo{booktitle}{\emph{Proceedings of the 19th International
  Conference on World Wide Web}} (Raleigh, North Carolina, USA)
  \emph{(\bibinfo{series}{WWW '10})}. \bibinfo{publisher}{ACM},
  \bibinfo{address}{New York, NY, USA}, \bibinfo{pages}{591--600}.
\newblock
\showISBNx{978-1-60558-799-8}
\urldef\tempurl%
\url{https://doi.org/10.1145/1772690.1772751}
\showDOI{\tempurl}


\bibitem[Leskovec(2013)]%
        {SNAP}
\bibfield{author}{\bibinfo{person}{J. Leskovec}.}
  \bibinfo{year}{2013}\natexlab{}.
\newblock \bibinfo{title}{SNAP: Stanford Network Analysis Platform}.
\newblock
\newblock
\urldef\tempurl%
\url{http://snap.stanford.edu/data/index.html}
\showURL{%
\tempurl}


\bibitem[Mawhirter et~al\mbox{.}(2021)]%
        {GraphZero}
\bibfield{author}{\bibinfo{person}{Daniel Mawhirter}, \bibinfo{person}{Sam
  Reinehr}, \bibinfo{person}{Connor Holmes}, \bibinfo{person}{Tongping Liu},
  {and} \bibinfo{person}{Bo Wu}.} \bibinfo{year}{2021}\natexlab{}.
\newblock \showarticletitle{GraphZero: A High-Performance Subgraph Matching
  System}.
\newblock \bibinfo{journal}{\emph{SIGOPS Oper. Syst. Rev.}}
  \bibinfo{volume}{55}, \bibinfo{number}{1} (\bibinfo{date}{June}
  \bibinfo{year}{2021}), \bibinfo{pages}{21–37}.
\newblock
\showISSN{0163-5980}
\urldef\tempurl%
\url{https://doi.org/10.1145/3469379.3469383}
\showDOI{\tempurl}


\bibitem[Mawhirter and Wu(2019)]%
        {AutoMine}
\bibfield{author}{\bibinfo{person}{Daniel Mawhirter} {and} \bibinfo{person}{Bo
  Wu}.} \bibinfo{year}{2019}\natexlab{}.
\newblock \showarticletitle{AutoMine: Harmonizing High-level Abstraction and
  High Performance for Graph Mining}. In \bibinfo{booktitle}{\emph{Proceedings
  of the 27th ACM Symposium on Operating Systems Principles}} (Huntsville,
  Ontario, Canada) \emph{(\bibinfo{series}{SOSP '19})}.
  \bibinfo{publisher}{ACM}, \bibinfo{address}{New York, NY, USA},
  \bibinfo{pages}{509--523}.
\newblock
\showISBNx{978-1-4503-6873-5}
\urldef\tempurl%
\url{https://doi.org/10.1145/3341301.3359633}
\showDOI{\tempurl}


\bibitem[Mhedhbi and Salihoglu(2019)]%
        {WCOJ}
\bibfield{author}{\bibinfo{person}{Amine Mhedhbi} {and} \bibinfo{person}{Semih
  Salihoglu}.} \bibinfo{year}{2019}\natexlab{}.
\newblock \showarticletitle{Optimizing Subgraph Queries by Combining Binary and
  Worst-Case Optimal Joins}.
\newblock \bibinfo{journal}{\emph{Proc. VLDB Endow.}} \bibinfo{volume}{12},
  \bibinfo{number}{11} (\bibinfo{date}{July} \bibinfo{year}{2019}),
  \bibinfo{pages}{1692–1704}.
\newblock
\showISSN{2150-8097}
\urldef\tempurl%
\url{https://doi.org/10.14778/3342263.3342643}
\showDOI{\tempurl}


\bibitem[Milo et~al\mbox{.}(2002)]%
        {Motifs2}
\bibfield{author}{\bibinfo{person}{R. Milo}, \bibinfo{person}{S. Shen-Orr},
  \bibinfo{person}{S. Itzkovitz}, \bibinfo{person}{N. Kashtan},
  \bibinfo{person}{D. Chklovskii}, {and} \bibinfo{person}{U. Alon}.}
  \bibinfo{year}{2002}\natexlab{}.
\newblock \showarticletitle{Network Motifs: Simple Building Blocks of Complex
  Networks}.
\newblock \bibinfo{journal}{\emph{Science}} \bibinfo{volume}{298},
  \bibinfo{number}{5594} (\bibinfo{year}{2002}), \bibinfo{pages}{824--827}.
\newblock
\showISSN{0036-8075}
\urldef\tempurl%
\url{https://doi.org/10.1126/science.298.5594.824}
\showDOI{\tempurl}
\showeprint{https://science.sciencemag.org/content/298/5594/824.full.pdf}


\bibitem[Pagh and Tsourakakis(2012)]%
        {ColorfulTC}
\bibfield{author}{\bibinfo{person}{Rasmus Pagh} {and}
  \bibinfo{person}{Charalampos~E. Tsourakakis}.}
  \bibinfo{year}{2012}\natexlab{}.
\newblock \showarticletitle{Colorful Triangle Counting and a MapReduce
  Implementation}.
\newblock \bibinfo{journal}{\emph{Inf. Process. Lett.}} \bibinfo{volume}{112},
  \bibinfo{number}{7} (\bibinfo{date}{mar} \bibinfo{year}{2012}),
  \bibinfo{pages}{277–281}.
\newblock
\showISSN{0020-0190}
\urldef\tempurl%
\url{https://doi.org/10.1016/j.ipl.2011.12.007}
\showDOI{\tempurl}


\bibitem[Paramonov et~al\mbox{.}(2019)]%
        {Lifting}
\bibfield{author}{\bibinfo{person}{Kirill Paramonov}, \bibinfo{person}{Dmitry
  Shemetov}, {and} \bibinfo{person}{James Sharpnack}.}
  \bibinfo{year}{2019}\natexlab{}.
\newblock \showarticletitle{Estimating Graphlet Statistics via Lifting}. In
  \bibinfo{booktitle}{\emph{Proceedings of the 25th ACM SIGKDD International
  Conference on Knowledge Discovery \& Data Mining}} (Anchorage, AK, USA)
  \emph{(\bibinfo{series}{KDD '19})}. \bibinfo{publisher}{Association for
  Computing Machinery}, \bibinfo{address}{New York, NY, USA},
  \bibinfo{pages}{587–595}.
\newblock
\showISBNx{9781450362016}
\urldef\tempurl%
\url{https://doi.org/10.1145/3292500.3330995}
\showDOI{\tempurl}


\bibitem[Pavan et~al\mbox{.}(2013)]%
        {NeighborhoodSampling}
\bibfield{author}{\bibinfo{person}{A. Pavan}, \bibinfo{person}{Kanat
  Tangwongsan}, \bibinfo{person}{Srikanta Tirthapura}, {and}
  \bibinfo{person}{Kun-Lung Wu}.} \bibinfo{year}{2013}\natexlab{}.
\newblock \showarticletitle{Counting and Sampling Triangles from a Graph
  Stream}.
\newblock \bibinfo{journal}{\emph{Proc. VLDB Endow.}} \bibinfo{volume}{6},
  \bibinfo{number}{14} (\bibinfo{date}{sep} \bibinfo{year}{2013}),
  \bibinfo{pages}{1870–1881}.
\newblock
\showISSN{2150-8097}
\urldef\tempurl%
\url{https://doi.org/10.14778/2556549.2556569}
\showDOI{\tempurl}


\bibitem[Pinar et~al\mbox{.}(2017)]%
        {ESCAPE}
\bibfield{author}{\bibinfo{person}{Ali Pinar}, \bibinfo{person}{C. Seshadhri},
  {and} \bibinfo{person}{Vaidyanathan Vishal}.}
  \bibinfo{year}{2017}\natexlab{}.
\newblock \showarticletitle{ESCAPE: Efficiently Counting All 5-Vertex
  Subgraphs}. In \bibinfo{booktitle}{\emph{Proceedings of the 26th
  International Conference on World Wide Web}} (Perth, Australia)
  \emph{(\bibinfo{series}{WWW '17})}. \bibinfo{publisher}{International World
  Wide Web Conferences Steering Committee}, \bibinfo{address}{Republic and
  Canton of Geneva, Switzerland}, \bibinfo{pages}{1431--1440}.
\newblock
\showISBNx{978-1-4503-4913-0}
\urldef\tempurl%
\url{https://doi.org/10.1145/3038912.3052597}
\showDOI{\tempurl}


\bibitem[Preti et~al\mbox{.}(2021)]%
        {MaNIACS}
\bibfield{author}{\bibinfo{person}{Giulia Preti}, \bibinfo{person}{Gianmarco
  De~Francisci~Morales}, {and} \bibinfo{person}{Matteo Riondato}.}
  \bibinfo{year}{2021}\natexlab{}.
\newblock \showarticletitle{MaNIACS: Approximate Mining of Frequent Subgraph
  Patterns through Sampling}. In \bibinfo{booktitle}{\emph{Proceedings of the
  27th ACM SIGKDD Conference on Knowledge Discovery \& Data Mining}} (Virtual
  Event, Singapore) \emph{(\bibinfo{series}{KDD '21})}.
  \bibinfo{publisher}{Association for Computing Machinery},
  \bibinfo{address}{New York, NY, USA}, \bibinfo{pages}{1348–1358}.
\newblock
\showISBNx{9781450383325}
\urldef\tempurl%
\url{https://doi.org/10.1145/3447548.3467344}
\showDOI{\tempurl}


\bibitem[Purohit et~al\mbox{.}(2017)]%
        {ASGS}
\bibfield{author}{\bibinfo{person}{Sumit Purohit}, \bibinfo{person}{Sutanay
  Choudhury}, {and} \bibinfo{person}{Lawrence~B. Holder}.}
  \bibinfo{year}{2017}\natexlab{}.
\newblock \showarticletitle{Application-specific graph sampling for frequent
  subgraph mining and community detection}. In \bibinfo{booktitle}{\emph{2017
  IEEE International Conference on Big Data (Big Data)}}.
  \bibinfo{pages}{1000--1005}.
\newblock
\urldef\tempurl%
\url{https://doi.org/10.1109/BigData.2017.8258022}
\showDOI{\tempurl}


\bibitem[Rossi et~al\mbox{.}(2019)]%
        {EgonetMotif}
\bibfield{author}{\bibinfo{person}{Ryan~A. Rossi}, \bibinfo{person}{Rong Zhou},
  {and} \bibinfo{person}{Nesreen~K. Ahmed}.} \bibinfo{year}{2019}\natexlab{}.
\newblock \showarticletitle{Estimation of Graphlet Counts in Massive Networks}.
\newblock \bibinfo{journal}{\emph{IEEE Transactions on Neural Networks and
  Learning Systems}} \bibinfo{volume}{30}, \bibinfo{number}{1}
  (\bibinfo{year}{2019}), \bibinfo{pages}{44--57}.
\newblock
\urldef\tempurl%
\url{https://doi.org/10.1109/TNNLS.2018.2826529}
\showDOI{\tempurl}


\bibitem[Sanei-Mehri et~al\mbox{.}(2018)]%
        {ButterflyCounting}
\bibfield{author}{\bibinfo{person}{Seyed-Vahid Sanei-Mehri},
  \bibinfo{person}{Ahmet~Erdem Sariyuce}, {and} \bibinfo{person}{Srikanta
  Tirthapura}.} \bibinfo{year}{2018}\natexlab{}.
\newblock \showarticletitle{Butterfly Counting in Bipartite Networks}. In
  \bibinfo{booktitle}{\emph{Proceedings of the 24th ACM SIGKDD International
  Conference on Knowledge Discovery \& Data Mining}} (London, United Kingdom)
  \emph{(\bibinfo{series}{KDD '18})}. \bibinfo{publisher}{Association for
  Computing Machinery}, \bibinfo{address}{New York, NY, USA},
  \bibinfo{pages}{2150–2159}.
\newblock
\showISBNx{9781450355520}
\urldef\tempurl%
\url{https://doi.org/10.1145/3219819.3220097}
\showDOI{\tempurl}


\bibitem[Shi et~al\mbox{.}(2020a)]%
        {ParallelCC}
\bibfield{author}{\bibinfo{person}{Jessica Shi}, \bibinfo{person}{Laxman
  Dhulipala}, {and} \bibinfo{person}{Julian Shun}.}
  \bibinfo{year}{2020}\natexlab{a}.
\newblock \showarticletitle{Parallel Clique Counting and Peeling Algorithms}.
  In \bibinfo{booktitle}{\emph{Conference on Applied and Computational Discrete
  Algorithms}}.
\newblock


\bibitem[Shi et~al\mbox{.}(2022)]%
        {Five-Cycle}
\bibfield{author}{\bibinfo{person}{Jessica Shi}, \bibinfo{person}{Louisa~Ruixue
  Huang}, {and} \bibinfo{person}{Julian Shun}.}
  \bibinfo{year}{2022}\natexlab{}.
\newblock \showarticletitle{Parallel Five-Cycle Counting Algorithms}.
\newblock \bibinfo{journal}{\emph{ACM J. Exp. Algorithmics}}
  \bibinfo{volume}{27}, Article \bibinfo{articleno}{4.1} (\bibinfo{date}{oct}
  \bibinfo{year}{2022}), \bibinfo{numpages}{23}~pages.
\newblock
\showISSN{1084-6654}
\urldef\tempurl%
\url{https://doi.org/10.1145/3556541}
\showDOI{\tempurl}


\bibitem[Shi et~al\mbox{.}(2020b)]%
        {GraphPi}
\bibfield{author}{\bibinfo{person}{Tianhui Shi}, \bibinfo{person}{Mingshu
  Zhai}, \bibinfo{person}{Yi Xu}, {and} \bibinfo{person}{Jidong Zhai}.}
  \bibinfo{year}{2020}\natexlab{b}.
\newblock \showarticletitle{GraphPi: High Performance Graph Pattern Matching
  through Effective Redundancy Elimination}. In
  \bibinfo{booktitle}{\emph{Proceedings of the International Conference for
  High Performance Computing, Networking, Storage and Analysis}} (Atlanta,
  Georgia) \emph{(\bibinfo{series}{SC '20})}. \bibinfo{publisher}{IEEE Press},
  Article \bibinfo{articleno}{100}, \bibinfo{numpages}{14}~pages.
\newblock
\showISBNx{9781728199986}


\bibitem[Slota and Madduri(2013)]%
        {fascia}
\bibfield{author}{\bibinfo{person}{George~M. Slota} {and}
  \bibinfo{person}{Kamesh Madduri}.} \bibinfo{year}{2013}\natexlab{}.
\newblock \showarticletitle{Fast Approximate Subgraph Counting and
  Enumeration}. In \bibinfo{booktitle}{\emph{2013 42nd International Conference
  on Parallel Processing}}. \bibinfo{pages}{210--219}.
\newblock
\urldef\tempurl%
\url{https://doi.org/10.1109/ICPP.2013.30}
\showDOI{\tempurl}


\bibitem[Slota and Madduri(2015)]%
        {CC-Parallel}
\bibfield{author}{\bibinfo{person}{George~M. Slota} {and}
  \bibinfo{person}{Kamesh Madduri}.} \bibinfo{year}{2015}\natexlab{}.
\newblock \showarticletitle{Parallel color-coding}.
\newblock \bibinfo{journal}{\emph{Parallel Comput.}}  \bibinfo{volume}{47}
  (\bibinfo{year}{2015}), \bibinfo{pages}{51--69}.
\newblock
\showISSN{0167-8191}
\urldef\tempurl%
\url{https://doi.org/10.1016/j.parco.2015.02.004}
\showDOI{\tempurl}
\newblock
\shownote{Graph analysis for scientific discovery}.


\bibitem[Spielman and Srivastava(2011)]%
        {gspar-resis}
\bibfield{author}{\bibinfo{person}{Daniel~A. Spielman} {and}
  \bibinfo{person}{Nikhil Srivastava}.} \bibinfo{year}{2011}\natexlab{}.
\newblock \showarticletitle{Graph Sparsification by Effective Resistances}.
\newblock \bibinfo{journal}{\emph{SIAM J. Comput.}} \bibinfo{volume}{40},
  \bibinfo{number}{6} (\bibinfo{year}{2011}), \bibinfo{pages}{1913--1926}.
\newblock
\urldef\tempurl%
\url{https://doi.org/10.1137/080734029}
\showDOI{\tempurl}
\showeprint{https://doi.org/10.1137/080734029}


\bibitem[Spielman and Teng(2004)]%
        {gpar-gspar}
\bibfield{author}{\bibinfo{person}{Daniel~A Spielman} {and}
  \bibinfo{person}{Shang-Hua Teng}.} \bibinfo{year}{2004}\natexlab{}.
\newblock \showarticletitle{Nearly-linear time algorithms for graph
  partitioning, graph sparsification, and solving linear systems}. In
  \bibinfo{booktitle}{\emph{Proceedings of the thirty-sixth annual ACM
  symposium on Theory of computing}}. \bibinfo{pages}{81--90}.
\newblock


\bibitem[Spielman and Teng(2011)]%
        {gspar-spectral}
\bibfield{author}{\bibinfo{person}{Daniel~A Spielman} {and}
  \bibinfo{person}{Shang-Hua Teng}.} \bibinfo{year}{2011}\natexlab{}.
\newblock \showarticletitle{Spectral sparsification of graphs}.
\newblock \bibinfo{journal}{\emph{SIAM J. Comput.}} \bibinfo{volume}{40},
  \bibinfo{number}{4} (\bibinfo{year}{2011}), \bibinfo{pages}{981--1025}.
\newblock


\bibitem[Teixeira et~al\mbox{.}(2015)]%
        {Arabesque}
\bibfield{author}{\bibinfo{person}{Carlos H.~C. Teixeira},
  \bibinfo{person}{Alexandre~J. Fonseca}, \bibinfo{person}{Marco Serafini},
  \bibinfo{person}{Georgos Siganos}, \bibinfo{person}{Mohammed~J. Zaki}, {and}
  \bibinfo{person}{Ashraf Aboulnaga}.} \bibinfo{year}{2015}\natexlab{}.
\newblock \showarticletitle{Arabesque: A System for Distributed Graph Mining}.
  In \bibinfo{booktitle}{\emph{Proceedings of the 25th Symposium on Operating
  Systems Principles}} (Monterey, California) \emph{(\bibinfo{series}{SOSP
  '15})}. \bibinfo{publisher}{ACM}, \bibinfo{address}{New York, NY, USA},
  \bibinfo{pages}{425--440}.
\newblock
\showISBNx{978-1-4503-3834-9}
\urldef\tempurl%
\url{https://doi.org/10.1145/2815400.2815410}
\showDOI{\tempurl}


\bibitem[Tsourakakis et~al\mbox{.}(2011)]%
        {SpectralTC}
\bibfield{author}{\bibinfo{person}{Charalampos~E Tsourakakis},
  \bibinfo{person}{Petros Drineas}, \bibinfo{person}{Eirinaios Michelakis},
  \bibinfo{person}{Ioannis Koutis}, {and} \bibinfo{person}{Christos
  Faloutsos}.} \bibinfo{year}{2011}\natexlab{}.
\newblock \showarticletitle{Spectral counting of triangles via element-wise
  sparsification and triangle-based link recommendation}.
\newblock \bibinfo{journal}{\emph{Social Network Analysis and Mining}}
  \bibinfo{volume}{1} (\bibinfo{year}{2011}), \bibinfo{pages}{75--81}.
\newblock


\bibitem[Tsourakakis et~al\mbox{.}(2009)]%
        {DOULION}
\bibfield{author}{\bibinfo{person}{Charalampos~E. Tsourakakis},
  \bibinfo{person}{U. Kang}, \bibinfo{person}{Gary~L. Miller}, {and}
  \bibinfo{person}{Christos Faloutsos}.} \bibinfo{year}{2009}\natexlab{}.
\newblock \showarticletitle{DOULION: Counting Triangles in Massive Graphs with
  a Coin}. In \bibinfo{booktitle}{\emph{Proceedings of the 15th ACM SIGKDD
  International Conference on Knowledge Discovery and Data Mining}} (Paris,
  France) \emph{(\bibinfo{series}{KDD '09})}. \bibinfo{publisher}{Association
  for Computing Machinery}, \bibinfo{address}{New York, NY, USA},
  \bibinfo{pages}{837–846}.
\newblock
\showISBNx{9781605584959}
\urldef\tempurl%
\url{https://doi.org/10.1145/1557019.1557111}
\showDOI{\tempurl}


\bibitem[Turk and Turkoglu(2019)]%
        {RevisitTC}
\bibfield{author}{\bibinfo{person}{Ata Turk} {and} \bibinfo{person}{Duru
  Turkoglu}.} \bibinfo{year}{2019}\natexlab{}.
\newblock \showarticletitle{Revisiting Wedge Sampling for Triangle Counting}.
  In \bibinfo{booktitle}{\emph{The World Wide Web Conference}} (San Francisco,
  CA, USA) \emph{(\bibinfo{series}{WWW '19})}. \bibinfo{publisher}{Association
  for Computing Machinery}, \bibinfo{address}{New York, NY, USA},
  \bibinfo{pages}{1875–1885}.
\newblock
\showISBNx{9781450366748}
\urldef\tempurl%
\url{https://doi.org/10.1145/3308558.3313534}
\showDOI{\tempurl}


\bibitem[\underline{Xuhao Chen} and Arvind(fpdf)]%
        {G2Miner}
\bibfield{author}{\bibinfo{person}{\underline{Xuhao Chen}} {and}
  \bibinfo{person}{Arvind}.} \bibinfo{year}{2022
  \href{https://arxiv.org/pdf/2112.09761.pdf}{[pdf]}}\natexlab{}.
\newblock \showarticletitle{{Efficient and Scalable Graph Pattern Mining on
  GPUs}}. In \bibinfo{booktitle}{\emph{Proceedings of the 16th USENIX Symposium
  on Operating Systems Design and Implementation {\bf (OSDI)}}}.
\newblock


\bibitem[Wang et~al\mbox{.}(2018)]%
        {RStream}
\bibfield{author}{\bibinfo{person}{Kai Wang}, \bibinfo{person}{Zhiqiang Zuo},
  \bibinfo{person}{John Thorpe}, \bibinfo{person}{Tien~Quang Nguyen}, {and}
  \bibinfo{person}{Guoqing~Harry Xu}.} \bibinfo{year}{2018}\natexlab{}.
\newblock \showarticletitle{RStream: Marrying Relational Algebra with Streaming
  for Efficient Graph Mining on a Single Machine}. In
  \bibinfo{booktitle}{\emph{Proceedings of the 12th USENIX Conference on
  Operating Systems Design and Implementation}} (Carlsbad, CA, USA)
  \emph{(\bibinfo{series}{OSDI'18})}. \bibinfo{publisher}{USENIX Association},
  \bibinfo{address}{Berkeley, CA, USA}, \bibinfo{pages}{763--782}.
\newblock
\showISBNx{978-1-931971-47-8}
\urldef\tempurl%
\url{http://dl.acm.org/citation.cfm?id=3291168.3291225}
\showURL{%
\tempurl}


\bibitem[Wang et~al\mbox{.}(2014)]%
        {PSRW}
\bibfield{author}{\bibinfo{person}{Pinghui Wang}, \bibinfo{person}{John C.~S.
  Lui}, \bibinfo{person}{Bruno Ribeiro}, \bibinfo{person}{Don Towsley},
  \bibinfo{person}{Junzhou Zhao}, {and} \bibinfo{person}{Xiaohong Guan}.}
  \bibinfo{year}{2014}\natexlab{}.
\newblock \showarticletitle{Efficiently Estimating Motif Statistics of Large
  Networks}.
\newblock \bibinfo{journal}{\emph{ACM Trans. Knowl. Discov. Data}}
  \bibinfo{volume}{9}, \bibinfo{number}{2}, Article \bibinfo{articleno}{8}
  (\bibinfo{date}{sep} \bibinfo{year}{2014}), \bibinfo{numpages}{27}~pages.
\newblock
\showISSN{1556-4681}
\urldef\tempurl%
\url{https://doi.org/10.1145/2629564}
\showDOI{\tempurl}


\bibitem[Yang and Leskovec(2012)]%
        {friendster}
\bibfield{author}{\bibinfo{person}{Jaewon Yang} {and} \bibinfo{person}{Jure
  Leskovec}.} \bibinfo{year}{2012}\natexlab{}.
\newblock \showarticletitle{Defining and Evaluating Network Communities based
  on Ground-truth}.
\newblock \bibinfo{journal}{\emph{CoRR}}  \bibinfo{volume}{abs/1205.6233}
  (\bibinfo{year}{2012}).
\newblock
\showeprint[arxiv]{1205.6233}
\urldef\tempurl%
\url{http://arxiv.org/abs/1205.6233}
\showURL{%
\tempurl}


\bibitem[Ye et~al\mbox{.}(2022)]%
        {LightningKclique}
\bibfield{author}{\bibinfo{person}{Xiaowei Ye}, \bibinfo{person}{Rong-Hua Li},
  \bibinfo{person}{Qiangqiang Dai}, \bibinfo{person}{Hongzhi Chen}, {and}
  \bibinfo{person}{Guoren Wang}.} \bibinfo{year}{2022}\natexlab{}.
\newblock \showarticletitle{Lightning Fast and Space Efficient K-Clique
  Counting}. In \bibinfo{booktitle}{\emph{Proceedings of the ACM Web Conference
  2022}} (Virtual Event, Lyon, France) \emph{(\bibinfo{series}{WWW '22})}.
  \bibinfo{publisher}{Association for Computing Machinery},
  \bibinfo{address}{New York, NY, USA}, \bibinfo{pages}{1191–1202}.
\newblock
\showISBNx{9781450390965}
\urldef\tempurl%
\url{https://doi.org/10.1145/3485447.3512167}
\showDOI{\tempurl}


\bibitem[Zhao et~al\mbox{.}(2010)]%
        {ParSE}
\bibfield{author}{\bibinfo{person}{Zhao Zhao}, \bibinfo{person}{Maleq Khan},
  \bibinfo{person}{V.~S.~Anil Kumar}, {and} \bibinfo{person}{Madhav~V.
  Marathe}.} \bibinfo{year}{2010}\natexlab{}.
\newblock \showarticletitle{Subgraph Enumeration in Large Social Contact
  Networks Using Parallel Color Coding and Streaming}. In
  \bibinfo{booktitle}{\emph{2010 39th International Conference on Parallel
  Processing}}. \bibinfo{pages}{594--603}.
\newblock
\urldef\tempurl%
\url{https://doi.org/10.1109/ICPP.2010.67}
\showDOI{\tempurl}


\bibitem[Zhao et~al\mbox{.}(2012)]%
        {SAHAD}
\bibfield{author}{\bibinfo{person}{Zhao Zhao}, \bibinfo{person}{Guanying Wang},
  \bibinfo{person}{Ali~R. Butt}, \bibinfo{person}{Maleq Khan},
  \bibinfo{person}{V.S.~Anil Kumar}, {and} \bibinfo{person}{Madhav~V.
  Marathe}.} \bibinfo{year}{2012}\natexlab{}.
\newblock \showarticletitle{SAHAD: Subgraph Analysis in Massive Networks Using
  Hadoop}. In \bibinfo{booktitle}{\emph{2012 IEEE 26th International Parallel
  and Distributed Processing Symposium}}. \bibinfo{pages}{390--401}.
\newblock
\urldef\tempurl%
\url{https://doi.org/10.1109/IPDPS.2012.44}
\showDOI{\tempurl}


\bibitem[Zhu et~al\mbox{.}(2023)]%
        {Arya}
\bibfield{author}{\bibinfo{person}{Zeying Zhu}, \bibinfo{person}{Kan Wu}, {and}
  \bibinfo{person}{Zaoxing Liu}.} \bibinfo{year}{2023}\natexlab{}.
\newblock \showarticletitle{Arya: Arbitrary Graph Pattern Mining with
  Decomposition-based Sampling}. In \bibinfo{booktitle}{\emph{Proceedings of
  the 20th USENIX Symposium on Networked Systems Design and Implementation}}
  \emph{(\bibinfo{series}{NSDI'23})}.
\newblock


\end{thebibliography}

\newpage

\newpage

\end{document}